\documentclass[aps,prl,reprint,superscriptaddress,a4paper]{revtex4-2}

\usepackage{amsmath, amsthm, amssymb}

\usepackage{mathtools}
\usepackage{physics}
\usepackage[pdftex]{graphicx}
\usepackage{times}
\usepackage{subfigure}
\usepackage{hhline}
\usepackage{mathrsfs}
\usepackage{color,soul}
\usepackage[normalem]{ulem}
\usepackage{natbib}
\usepackage{bbold}
\usepackage{placeins}
\usepackage{tcolorbox}
\usepackage{braket}
\usepackage{soul,xcolor}
\usepackage{tabu}
\usepackage{array}
\usepackage{longtable}
\usepackage{multirow}
\usepackage{tikz}

\usepackage{bm}

\usepackage[utf8]{inputenc}

\definecolor{myurlcolor}{rgb}{0,0,0.7}
\definecolor{myrefcolor}{rgb}{0.8,0,0}
\usepackage[unicode=true,pdfusetitle, bookmarks=false,bookmarksnumbered=false,
bookmarksopen=false, breaklinks=false,pdfborder={0 0 0},backref=false,
colorlinks=true, linkcolor=myrefcolor,citecolor=myurlcolor,urlcolor=myurlcolor]{hyperref}

\DeclareGraphicsExtensions{.png,.pdf,.eps,.jpg}
\graphicspath{ {./figs/} }
\newtheorem{thm}{Theorem}

\newcommand{\PRLsection}[1]{\emph{#1.---}}

\newcommand{\eref}[1]{(\ref{#1})}
\newcommand{\eqnref}[1]{Eq.~(\ref{#1})}
\newcommand{\eqnsref}[2]{Eqs.~(\ref{#1}) and (\ref{#2})}
\newcommand{\figref}[1]{Fig.~\ref{#1}}
\newcommand{\tabref}[1]{Table~\ref{#1}}

\newcommand{\appref}[1]{App.~\ref{#1}}

\newcommand{\citeref}[1]{Ref.~\cite{#1}}

\renewcommand\Re{\operatorname{Re}}
\renewcommand\Im{\operatorname{Im}}

\renewcommand\Tr{\operatorname{Tr}}

\newcolumntype{M}[1]{>{\centering\arraybackslash}m{#1}}
\newcolumntype{N}{@{}m{0pt}@{}}

\newcommand{\mrm}[1]{\mathrm{#1}}

\renewcommand{\d}{\mrm{d}}
\renewcommand{\i}{\mrm{i}}
\newcommand{\e}{\mrm{e}}
\newcommand{\T}{{\mrm{T}}}

\newcommand{\qmean}[1]{{\langle#1\rangle}}

\newcommand{\mean}[1]{{\mathbb{E}\!\left[#1\right]}}
\newcommand{\meanP}[2]{{\mathbb{E}_{#1}\!\left[#2\right]}}

\newcommand{\MSE}[1]{\Delta^2{#1}}
\newcommand{\errormat}{{{\bm{\Delta}}^2 \tilde{\parVec}}}
\newcommand{\errorM}{\Delta_{*}}
\newcommand{\cerror}{\Delta_\mrm{C}}
\newcommand{\qerror}{\Delta_\mrm{Q}}
\newcommand{\errorS}{\delta_{*}}
\newcommand{\cerrorS}{\delta_\mrm{C}}
\newcommand{\qerrorS}{\delta_\mrm{Q}}

\newcommand{\mat}[1]{{\textbf{#1}}}

\newcommand{\qvec}[1]{\hat{\bm{#1}}}
\newcommand{\id}{\mat{I}}   

\renewcommand{\H}{{\mat{H}}}

\newcommand{\HS}{\H_{\theta_1}^\mrm{S}}
\newcommand{\HNS}{\H_{\theta_1}^\mrm{NS}}

\newcommand{\smat}[1]{{\mathsf{#1}}}
\newcommand{\sformletter}{J}
\newcommand{\sform}{{\smat{\sformletter}}} 
\newcommand{\stransf}[1]{{\mathcal{\sformletter}\!\left\{#1\right\}}}

\newcommand{\svec}[1]{{\smat{#1}}}
\newcommand{\qsvec}[1]{{\hat{\smat{#1}}}}
\newcommand{\Id}{\smat{I}}	 

\newcommand{\dprime}{{\prime \prime}}

\newcommand{\parVec}{{\bm{\theta}}}
\renewcommand{\pmat}[1]{{\bm{#1}}}

\newcommand{\F}{{\pmat{F}}}
\newcommand{\FQ}{{\pmat{\mathcal{F}}}}

\newtheorem{lemma}{Lemma}

\newcommand{\X}{{\smat{X}}}
\newcommand{\A}{{\smat{A}}}

\newcommand{\beq}{\begin{equation}}
\newcommand{\eeq}{\end{equation}}

\newcommand{\bpm}{\begin{pmatrix}}
	\newcommand{\epm}{\end{pmatrix}}

\newcommand{\cP}{\mathcal{P}}

\newcommand*\circled[1]{\tikz[baseline=(char.base)]{
		\node[shape=circle,draw,inner sep=2pt] (char) {#1};}}

\makeatletter

\begin{document}

\title{Multiparameter estimation perspective on non-Hermitian singularity-enhanced sensing}

\author{Javid Naikoo}
\email{j.naikoo@cent.uw.edu.pl}
\affiliation{Centre for Quantum Optical Technologies, Centre of New Technologies, University of Warsaw, Banacha 2c, 02-097 Warszawa, Poland}

\author{Ravindra W. Chhajlany}
\affiliation{Institute of Spintronics and Quantum Information, Faculty of Physics, Adam Mickiewicz University, 61-614 Pozna\'{n}, Poland}
\affiliation{Centre for Quantum Optical Technologies, Centre of New Technologies, University of Warsaw, Banacha 2c, 02-097 Warszawa, Poland}

\author{Jan Ko\l{}ody\'{n}ski}
\email{jan.kolodynski@cent.uw.edu.pl}
\affiliation{Centre for Quantum Optical Technologies, Centre of New Technologies, University of Warsaw, Banacha 2c, 02-097 Warszawa, Poland}

\begin{abstract}
Describing the evolution of quantum systems by means of non-Hermitian generators opens a new avenue to explore the dynamical properties naturally emerging in such a picture, e.g.~operation at the so-called exceptional points, preservation of parity-time symmetry, or capitalising on the singular behaviour of the dynamics. In this work, we focus on the possibility of achieving unbounded sensitivity when using the system to sense linear perturbations away from a singular point. By combining multiparameter estimation theory of Gaussian quantum systems with the one of singular-matrix perturbations, we introduce the necessary tools to study the ultimate limits on the precision attained by such singularity-tuned sensors. We identify under what conditions and at what rate can the resulting sensitivity indeed diverge, in order to show that nuisance parameters should be generally included in the analysis, as their presence may alter the scaling of the error with the estimated parameter.
\end{abstract}

\maketitle

\PRLsection{Introduction}%
%
Quantum entanglement boosts dramatically performance in sensing~\cite{Degen2017,Pezze2018}, allowing quantum sensors to breach classical limits imposed by the i.i.d.-statistics~\cite{Giovannetti2004}. The corresponding enhancement, however, turns out to be very fragile~\cite{Maccone2011,Escher2011,Demkowicz2012}, making methods of quantum control~\cite{Chaves2013,len_quantum_2021,Kaoru2022} and error-correction~\cite{Sekatski2017,zhou_achieving_2018,Demkowicz2017} essential, if the robustness against decoherence and imperfections is to be maintained. As the impact of noise becomes inevitable with sensor complexity, a change of paradigm is necessary. One way is to adopt a non-Hermitian dynamical description and engineer the noise instead, in order to make the evolution extremely sensitive to external perturbations. For instance, by considering deviations from the \emph{exceptional points} (EPs) in the space of parameters characterising the system~\cite{WiersigSensing14}---special degeneracies at which $n$ (complex) eigenvalues coalesce along with their respective eigenmodes~\cite{Heiss2004,Miri2019,Ravi2019}---a linear perturbation $\epsilon$ away from the EP leads to an $n$th-root splitting $\sim\!\sqrt[n]{\epsilon}$ of the eigenmode frequencies~\cite{Wiersig20review} (in contrast to the polynomial energy-splittings arising when perturbing Hermitian Hamiltonians). Hence, a splitting measurement may yield infinitely steep signals of unbounded sensitivity as $\epsilon\to0$, as demonstrated with optical resonators~\cite{Chen2017,Hodaei2017} in the regime in which the measurement-induced noise can be ignored.
Otherwise, the effect is washed out by the quantum noise~\cite{hongkong,Wiersig2020,Wang2020}---in a similar way as it prohibits noiseless amplification of optical signals~\cite{Clerk2010,Tsang2018}. 

Alternative schemes involving linearly coupled systems were proposed (and implemented~\cite{Kononchuk2022}) that surpass the impact of quantum noise by resorting to different perturbations of the effective non-Hermitian generator, $\H$, in the Langevin formalism~\cite{Gardiner1985}---with the operation around an EP being no longer essential~\footnote{See~\cite{Chu2020} for considerations of also finite dimensional systems.}. For example, by considering the internal interaction to be non-reciprocal and perturbing the coupling strength instead, the sensitivity---the \emph{signal-to-noise ratio} (SNR)---was shown to improve by a constant factor~\cite{AC1}. Moreover, it was shown that by engineering $\H$ to be \emph{singular}
\footnote{\citeref{LJ} attributed this to EP and lasing-threshold conditions being simultaneously fulfilled, while actually satisfying the singularity condition~\cite{hongkong}.} and sensing perturbations of the internal frequency with the probing signal tuned to it, the SNR may diverge boundlessly as $\epsilon^{-2}$ with $\epsilon\to0$~\cite{LJ}. Despite the apparent similarity to the EP-induced effect, this is a consequence of probing the quantum sensor close to a dynamical phase transition, which constitutes a resource in sensing~\cite{Macieszczak2016,Fernandez2017,Wald2020,chu_dynamic_2021,wu_criticality-enhanced_2021,xie_quantum_2021,di_candia_critical_2021,Ilias2022}. Although the linearity of dynamics may be questioned at such a critical point~\cite{hongkong,Wiersig20review}, for the class of sensors we consider~\cite{Hodaei2014,Feng2014,Peng2014b,Peng2014,Partanen2019} it has been verified experimentally at different probe powers~\cite{Peng2014}. Here, we demonstrate how to correctly assess such singularity-tuned sensors~\cite{LJ}, and hence possibly other criticality-enhanced sensing schemes~\cite{Macieszczak2016,Fernandez2017,Wald2020,chu_dynamic_2021,wu_criticality-enhanced_2021,xie_quantum_2021,di_candia_critical_2021,Ilias2022}.

\begin{figure}
	\includegraphics[width=\columnwidth]{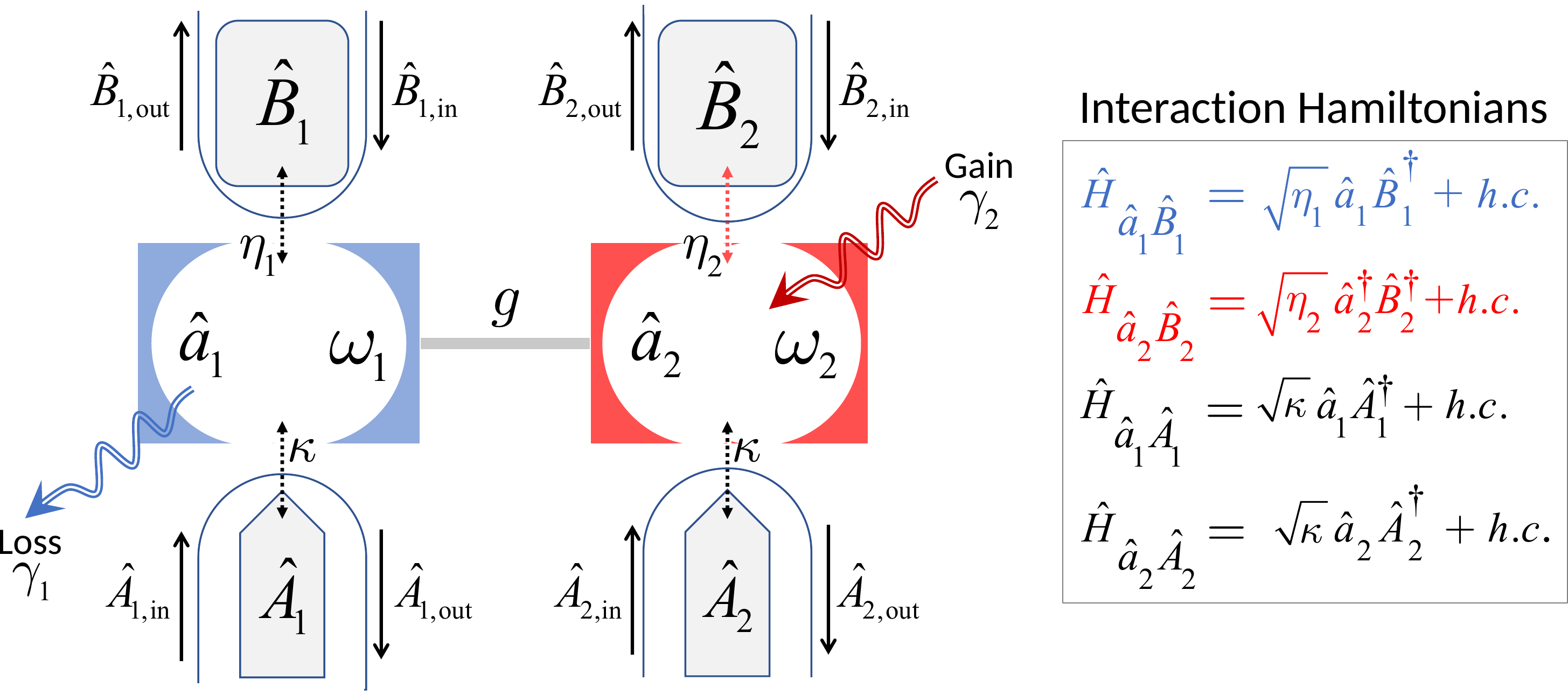}
	\caption{\textbf{Non-Hermitian sensor} model consisting of two coupled cavities bearing optical modes $\hat{a}_1$ and $\hat{a}_2$ with effective loss and gain rates $\gamma_1=\tfrac{\eta_1 + \kappa}{2}$ and $\gamma_2=\tfrac{\eta_2 - \kappa}{2}$, respectively. The rates are controlled by the `scattering channels' (modes), $\hat{B}_1$ and $\hat{B}_2$, coupled to the respective cavities, whose dynamics is also affected by the `probing channels', $\hat{A}_1$ and $\hat{A}_2$, used to continuously monitor each cavity.}
	\label{fig:model}
\end{figure}

In particular, we analyse the emergence of singularity-induced SNR-divergence within the canonical linear system exhibiting \emph{parity-time} (PT) \emph{symmetry}~\cite{El-Ganainy2018}, i.e.~two coupled bosonic cavities that experience loss and gain while being continuously monitored~\cite{Hodaei2014,Feng2014,Peng2014b,Peng2014,Partanen2019}---see \figref{fig:model}---whose complexity of non-Hermitian dynamics covers the settings of all:~EP\cite{Wiersig20review}-, non-reciprocity\cite{AC1}- and singularity\cite{LJ}-based sensing. We demonstrate how these approaches fit into a single picture, from which it follows that one must resort to multiparameter estimation theory of multimode Gaussian states~\cite{Nichols}, in order to determine the ultimate limits on sensitivity. We then perform \emph{singular perturbations}~\cite{Sain1969,Howlett1982,Howlett2001,avrachenkov2013analytic} of the corresponding frequency response function, in order to show that the SNR-divergence critically depends on the perturbation form, even when assuming all other system parameters to be known. Furthermore, not only any slight deviation from the singular point precludes unbounded sensitivity, but also the divergence rate depends on \emph{nuisance parameters}, i.e.~other system parameters unknown prior to estimation~\cite{Rafal2020,Suzuki2020}. Thus, one must be careful when assessing the sensing capabilities of singularity-tuned sensors~\cite{LJ}, as these depend strongly on the ability to fine-tune and calibrate the system.

\PRLsection{Non-Hermitian sensor model}%
Sensors~\cite{Hodaei2014,Feng2014,Peng2014b,Peng2014,Partanen2019} can be conveniently described by the model depicted in \figref{fig:model}, in which two cavities containing optical modes $\hat{a}_1$ and $\hat{a}_2$ at frequencies $\omega_1$ and $\omega_2$, respectively, are linearly coupled with strength $g$, so that the overall free Hamiltonian reads: 
\begin{align}
	\hat{H}_S &= \omega_1 \hat{a}_1^\dagger  \hat{a}_1 +  \omega_2 \hat{a}_2^\dagger  \hat{a}_2  + g \left(\hat{a}_1^\dagger \hat{a}_2 + \hat{a}_2^\dagger \hat{a}_1\right).
	\label{eq:HS}
\end{align}
We shall consider here the case of degenerate cavities such that $\omega_1=\omega_2\eqqcolon\omega_0$~\cite{Peng2014}. Each mode $\hat{a}_1$ ($\hat{a}_2$) is separately coupled to a scattering channel $\hat{B}_1$ ($\hat{B}_2$) that effectively induces loss (gain) of strength $\eta_1$ ($\eta_2$) on each cavity. The couplings correspond to effective  beam-splitter  and non-degenerate parametric-amplifier \cite{clerk2010introduction} interactions, see \figref{fig:model}, assuming that in the latter case other non-linear effects, e.g.~degenerate parametric amplification~\cite{Bruch:19}, are absent. In parallel, both cavities are independently probed via channels $\hat{A}_1$ and $\hat{A}_2$, each coupled with strength $\kappa$, whose outputs are continuously monitored. By resorting to the input-output formalism~\cite{GardinerBook,Guta2016} summarised in \appref{app:Langevin_eqs}, the sensor dynamics can be described by a linear quantum Langevin equation~\cite{Gardiner1985}:
\begin{equation}
	\partial_t 	\qvec{a} = -\i (\omega_0 \id + \H) \qvec{a} + \qvec{A}_\mrm{in} + 
	\qvec{B}_\mrm{in},
	\label{eq:dyn_model}
\end{equation}
where $\qvec{a}\coloneqq\{\hat{a}_1,\hat{a}_2\}^\T$, $\qvec{A}_\mrm{in}\coloneqq\{\sqrt{\kappa} \hat{A}_{1,\mrm{in}}, \sqrt{\kappa}\hat{A}_{2,\mrm{in}}\}^\T$,
$\qvec{B}_\mrm{in}\coloneqq\{\sqrt{\eta_1} \hat{B}_{1,\mrm{in}},- \sqrt{\eta_2} \hat{B}_{2,\mrm{in}}^\dagger\}^\T$, 
and $\id$ is a 2$\times$2 identity matrix. 
As depicted in \figref{fig:model}, $\hat{A}_{\ell, \mrm{in}}$,  $\hat{B}_{\ell, \mrm{in}}$ with $\ell=1,2$ denote the effective input fields of the optical channels, whose output fields are determined by the input-output relations~\cite{Gardiner1985} and read at time $t$:~$\hat{A}_{\ell,\mrm{out}} (t) = \hat{A}_{\ell,\mrm{in}}(t) - \sqrt{\kappa} \,\hat{a}_{\ell}(t)$, see \appref{app:Langevin_eqs}.

The inter-cavity interactions, as well as their coupling to optical channels, modify the free evolution of the cavity modes in \eqnref{eq:dyn_model} in the form of a non-Hermitian dynamical generator:
\begin{equation}
\H = \begin{pmatrix}
	 - \i \gamma_1 	&	g  \\
		g     &	 + \i \gamma_2              
	\end{pmatrix}
\label{eq:sys_H}
\end{equation}
with $\gamma_1\coloneqq(\eta_1 + \kappa)/2$ ($\gamma_2\coloneqq(\eta_2 - \kappa)/2$) being the overall loss (gain) rate of each cavity. As a result, defining $\gamma_\pm\coloneqq(\gamma_2 \pm \gamma_1)/2$, we may write the eigenvalues and (unnormalised) eigenmodes of $\H$ as
\begin{equation}
	\lambda_{\pm} =  \i \gamma_- \pm \sqrt{g^2 -\gamma_+^2}, 
	\left|e_{\pm}\right)\! = \!\begin{pmatrix} -\i\gamma_+ \pm \sqrt{g^2 - \gamma_+^2} \\ g \end{pmatrix},
	\label{eq:eval_evecs}
\end{equation}
so that it becomes clear that the spectrum of $\H$ is real iff $g\ge\gamma_+$ and $\gamma_-=0$, in which case $\H$ formally exhibits \emph{PT-symmetry}~\cite{El-Ganainy2018}. Such a condition is conveniently visualised in the 3D-space of parameters $\{\gamma_1,\gamma_2,g\}$, see \figref{fig:par_space}\textbf{a}, by a vertical (yellow) triangular plane. Importantly, it has been demonstrated that by maintaining the PT-symmetry, the validity of the linear model \eref{eq:dyn_model} can be extended to high probe powers~\cite{Peng2014}. The sole condition $\gamma_-=0$ we term as the balanced scenario, as the gain then balances out exactly the loss ($\gamma_1=\gamma_2$)~\cite{Peng2014b}. In what follows, when probing the system at the sensor frequency $\omega_0$, the \emph{singularity} of the non-Hermitian generator \eref{eq:sys_H} will play a pivotal role. This corresponds to the condition $\det\H=0$ or $g^2=\gamma_1 \gamma_2$, which is represented by the green surface in \figref{fig:par_space}\textbf{a}. In contrast, the generator \eref{eq:sys_H} exhibits an EP when gain and loss are such that $g=\gamma_+$, characterized by the simultaneous coalescence of the two eigenvalues $\lambda_\pm$ and eigenmodes $\left|e_\pm\right)$~\cite{El-Ganainy2018}. In \figref{fig:par_space}\textbf{a}, however, we mark the EP-condition only in the relevant case, i.e.~when the singularity is simultaneously fulfilled (dashed blue line). Now, as emphasised by the constant-$\gamma_1$ cut of the parameter-space in \figref{fig:par_space}\textbf{b}, the singularity and PT-symmetry conditions separate the regions of $\mathrm{Im}\lambda_\pm$ being either negative or positive, which yield the dynamics \eref{eq:dyn_model} either stable or unstable, respectively~\footnote{Manifested by the absence of steady-state solution, and the cavity occupation numbers becoming divergent as $t\to\infty$.}---the latter coloured in grey in \figref{fig:par_space}. The border $\mathrm{Im}\lambda_\pm=0$ defines then the \emph{lasing threshold}~\cite{Peng2014b}, which we mark for $\gamma_1=0.75$ in \figref{fig:par_space} by the dashed purple line. Although we defer the details to \appref{app:lasing_threshold} and \cite{NewProject}, let us note that for $g>\gamma_1$ the system is at the lasing threshold when balanced, while the singularity is then not fulfilled.

\begin{figure}
	\includegraphics[width=\columnwidth]{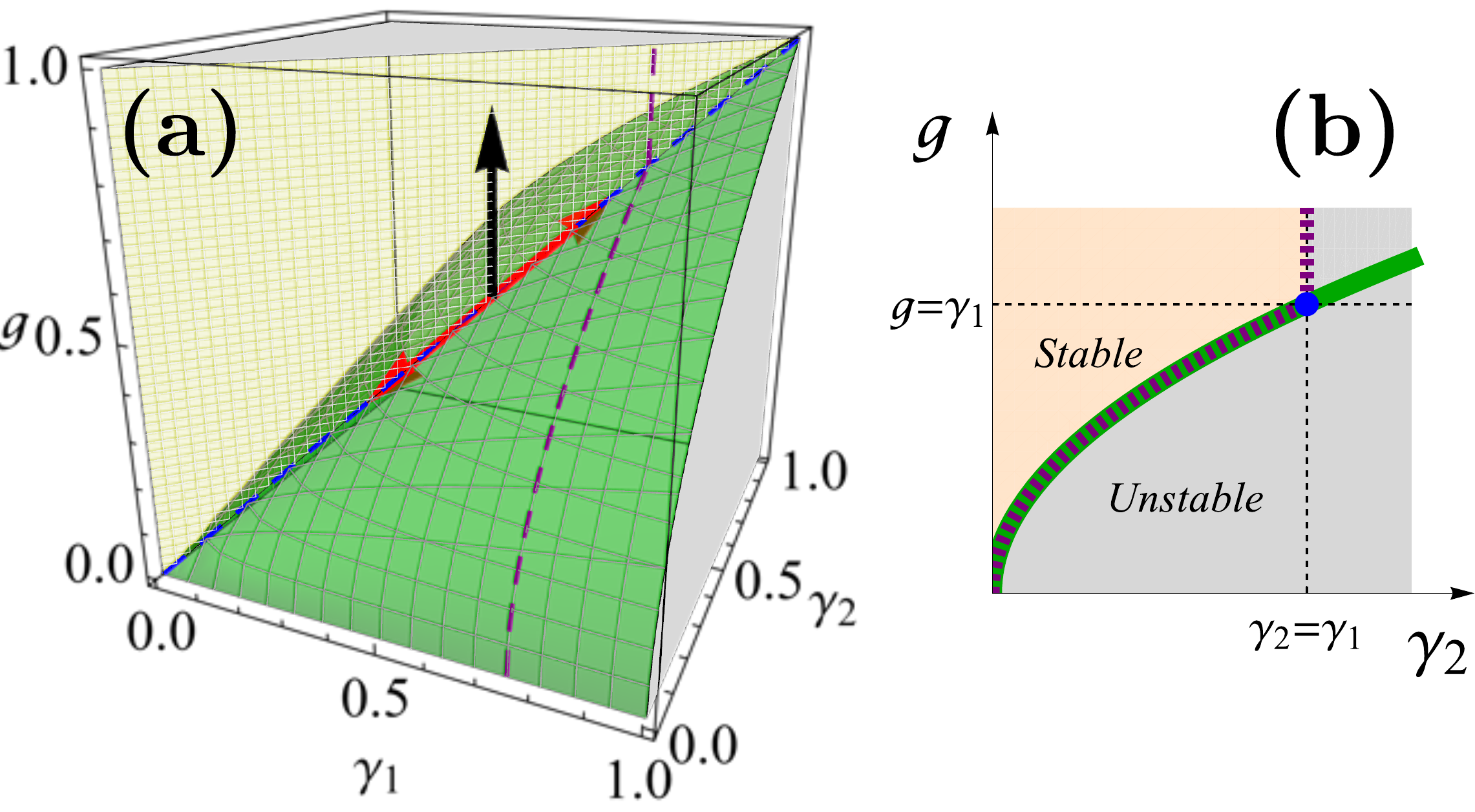}
	\caption{\textbf{Space of parameters characterising the dynamical generator} $\H$ in \eqnref{eq:sys_H}, with $g$ and $\gamma_{1}$($\gamma_{2}$) being the coupling constant and loss (gain) rates, respectively. \textbf{(a)}: The green surface contains the parameter values for which the \emph{singularity} condition $\det\H=0$ is fulfilled, while the triangular vertical plane ensures the \emph{PT-symmetry}. At their intersection, at which $g=\gamma_1=\gamma_2$, marked with a dashed blue line, the \emph{EP-condition} is guaranteed. \textbf{(b)}: A space-cut at exemplary $\gamma_1=0.75$, for which the separation between stable and unstable regions of dynamics is clearly visible---with the \emph{lasing threshold} (purple dashed line) following the singularity surface before becoming (at an EP, blue dot) equivalent to the PT-condition for $g>\gamma_1$. In the multiparameter estimation setting considered, we choose perturbations of the nuisance parameter to either invalidate or maintain the singularity condition, while always preserving the PT-symmetry---as marked with black and red arrows in \textbf{(a)}, respectively.  
	}
	\label{fig:par_space}
\end{figure}

In sensing tasks with linear perturbations, the generator \eref{eq:sys_H} is modified as follows:
\begin{equation}
	\H_\parVec\coloneqq  \H - \sum_{i = 0} ^{m} \theta_i \mat{n}_i = \H_{\bar{\parVec}} - \theta_0 \mat{n}_0,
	\label{eq:sens_H}
\end{equation}
where $\parVec \coloneqq \{\theta_i\}_i$ denotes a set of ($m+1$, real) parameters to be sensed, each of which modifies $\H$ according to some (complex) 2$\times$2 matrix $\mat{n}_i$. In \eqnref{eq:sens_H}, we single out the case in which $\theta_0$ denotes the \emph{primary} parameter to be sensed around zero,  while the rest of the set, $\bar{\parVec} \coloneqq \{\theta_i\}_{i\ne 0}$, contains \emph{nuisance} parameters, i.e.~ones that are of no interest but nonetheless unknown. This allows us to capture all the following $\theta_0$-estimation scenarios. By setting $\mat{n}_0=\sigma_z$ in \eqnref{eq:sens_H}, we let $\theta_0=(\omega_1-\omega_2)/2$ describe perturbations of the detuning between the cavity frequencies---as originally considered in the EP-based sensing schemes~\cite{hongkong}. By choosing $\mat{n}_0=\sigma_x$ instead, we let $\theta_0$ perturb the coupling strength $g$---as investigated in \citeref{AC1} dealing with non-reciprocal dynamics. Finally, we consider $\mat{n}_0=\id$ ($\mat{n}_0=(1,0;0,0)$), in which case $\theta_0$ describes perturbations of the common frequency $\omega_0$ (of $\omega_1$ for the first cavity only)~\cite{LJ}.

\PRLsection{Linear response in the Fourier domain}%
We consider the sensor to be interacting with Gaussian light~\footnote{However, we do not exclude the possibility of non-Gaussian measurements.}, so it is sufficient to describe its dynamics using the Gaussian formalism~\cite{Ferraro2005,Weedbrook2012};~within which evolution of any bosonic modes $\hat{b}_i$ is fully characterised after defining the vector $\qsvec{S}=\{\hat{q}_1, \hat{q}_2, \dots, \hat{p}_1, \hat{p}_2,\dots\}^\T$ of their quadratures, $\hat{q}_i=\hat{b}_i+\hat{b}_i^\dagger$ and $\hat{p}_i=-\i(\hat{b}_i-\hat{b}_i^\dagger)$, and tracking its mean $\svec{S}\coloneqq\qmean{\qsvec{S}}$ and covariance $\smat{V}$ with entries $\smat{V}_{jk}\coloneqq\frac{1}{2} \qmean{ \{\qsvec{S}_j, \qsvec{S}_k\}}  - \qmean{\qsvec{S}_j}\qmean{\qsvec{S}_k}$. As we are interested in probing the sensor at a particular frequency $\omega$, we focus on the evolution in the Fourier space, in which according to \eqnref{eq:dyn_model} the dynamics of measured outputs, $\qsvec{S}_\mrm{out}^A$ containing quadratures of $\hat{A}_{\ell,\mrm{out}}[\omega]\coloneqq\int\!\d t\,\e^{\i\omega t}\hat{A}_{\ell,\mrm{out}}(t)$, is given by $\svec{S}_\mrm{out}^A = \left(\Id - \kappa\smat{G} \right) \svec{S}_\mrm{in}^A$ and $\smat{V}_\mrm{out}^A  = \left(\Id - \kappa\smat{G} \right) \smat{V}_\mrm{in}^A \left(\Id - \kappa\smat{G} \right)^\T + \kappa \smat{G} \Xi \tilde{\smat{V}}_\mrm{in}^{B} \Xi^\T \smat{G}^\T$, where the covariance of probe outputs $\smat{V}_\mrm{out}^A$ depends also on the overall covariance of input scattering modes, i.e.~$\tilde{\smat{V}}_\mrm{in}^B$ describing correlations between 8 quadratures ($\ell=1,2$) of $\hat{B}_{\ell,\mrm{in}}[\pm\omega]$, see \appref{app:dyns_omega} for derivation. Above, $\Xi$ is the matrix associated with the coupling of cavities to scattering channels, $\Id$ denotes a 4$\times$4 identity matrix, while the central object is the \emph{(linear) response function}~\footnote{Equivalent to the transfer-function matrix arising when the linear dynamics is solved in the Laplace domain instead~\cite{Guta2016}.}:
\begin{equation}
	\smat{G}[\omega] = \sform \left((\omega-\omega_0)\Id - \smat{H}\right)^{-1},
	\label{eq:Greens}
\end{equation}
whose divergent behaviour, as shown, is responsible for the unbounded precision when sensing perturbations. By $\sform=(0,-\id;\id,0)$ we denote above the symplectic form consistent with the notation of \appref{app:dyns_omega}, within which~%
\footnote{Phase-space representation of a matrix $\mat{m}$ reads $\smat{m}\equiv\stransf{\mat{m}}\coloneqq(\Re[\mat{m}],- \Im[\mat{m}];\Im[\mat{m}],\Re[\mat{m}])$, e.g.~$\sform\coloneqq\stransf{\i\,\id}$.}:
\begin{equation}
\smat{H}\coloneqq 	\begin{pmatrix}
	        	\Re[\H]   & - \Im[\H] \\
	            \Im[\H]  & \Re[\H]
            \end{pmatrix}
        = 	\begin{pmatrix}
              0 & g & \gamma_1 & 0 \\
              g & 0 & 0 & -\gamma_2 \\
              -\gamma_1 & 0 & 0 & g \\
              0 & \gamma_2 & g & 0 
        	\end{pmatrix},
	\label{eq:H0_smatrix}
\end{equation}
is the phase-space representation of $\H$. Crucially, \eqnref{eq:Greens} diverges if $\det\!\left((\omega - \omega_{0})\Id - \smat{H}\right)=0$, which can always be assured by tuning loss and gain, as long as the probing is performed \emph{in resonance} ($\omega=\omega_0$) with the common internal frequency. Divergence then occurs when $\det\smat{H}={|\!\det\H|}^2=0$, i.e.~when the generator $\H$ in \eref{eq:sys_H} is, indeed, singular~\footnote{For $\omega\neq\omega_0$, the response function \eref{eq:Greens} can be made divergent only for a balanced system, $\gamma_1=\gamma_2=\gamma$, with strong coupling fixed to $g=\sqrt{\gamma^2+(\omega-\omega_0)^2}>\gamma$, see \appref{app:sense_at_LT}.}.

Considering now $\parVec$-parametrised perturbations specified in \eqnref{eq:sens_H}, the response function \eref{eq:Greens} becomes 
\begin{equation} 
	\smat{G}_{\parVec}[\omega=\omega_0] = \sform \left(\sum_{i = 0} ^{m} \theta_i \smat{n}_i - \smat{H}\right)^{-1}
	\!\!\! =
	\sform \left(\theta_0 \smat{n}_0-\smat{H}_{\bar{\parVec}} \right)^{-1},
	\label{eq:Gomega}
\end{equation}
where $\smat{n}_i$ are the phase-space representations of $\mat{n}_i$~\cite{Note6}. Analogously to \eqnref{eq:sens_H}, we highlight the case when $\theta_0$ is the primary parameter and $\smat{H}_{\bar{\parVec}}\coloneqq\smat{H}-\sum_{i \neq 0} \theta_i \smat{n}_i$. When considering singular dynamics  with $\det\smat{H}=0$, the response function \eref{eq:Gomega} diverges for all $\theta_i=0$ with the character of divergence specified by the matrices $\{\smat{n}_i\}_i$. For instance, when estimating only $\theta_0\approx0$ the singularity may be maintained despite $\bar{\parVec}\neq0$---as long as the condition $\det\smat{H}_{\bar{\parVec}}=0$ is assured also for $\bar{\parVec}\neq0$.

Aiming to study such subtleties and treat $\bar{\parVec}$ as nuisance parameters, we focus on the \emph{two-parameter} setting, $\parVec=\{\theta_0,\theta_1\}$, in which the primary parameter represents $\omega_0$-frequency perturbations, i.e.~$\mat{n}_0=\id$~\cite{LJ}. In contrast, we choose the secondary parameter $\theta_1$ such that its variations either invalidate or maintain the singularity condition, but do not break the PT-symmetry, so that the linearity of dynamics is ensured~\cite{Peng2014}. In \figref{fig:par_space}\textbf{a}, we mark these cases by black and red arrows, respectively, which correspond to \emph{singularity non-preserving} (NS) increase of the coupling $g$~\cite{AC1}, and \emph{singularity preserving} (S) perturbations maintaining $g=\gamma_1=\gamma_2$ and, hence, the EP-condition%
~\footnote{However, the maintenance of EP is not essential when not restricted to PT-preserving perturbations.}.
In particular, the two scenarios lead to \eqnref{eq:sens_H} of the form $\H_{\parVec}= \H_{\theta_1}^\mrm{NS/S} - \theta_0 \id$, where
\begin{align}
	\HNS  \coloneqq \bar{\H} - \theta_1 \sigma_x  
	\quad\text{and}\quad 
	\HS  \coloneqq \bar{\H} - \theta_1(\sigma_x - \i \sigma_z),
	\label{eq:HS_HNS}
\end{align}
and $\bar{\H}$ denotes $\H$ with $g=\gamma_1=\gamma_2=1$~%
\footnote{
Any value of $\chi>0$ such that $g=\gamma_1=\gamma_2=\chi$ (defining a point on the blue dashed line in \figref{fig:par_space}a), can be rescaled to $\chi=1$ by redefining the perturbation parameters $\theta_j$.}.

\PRLsection{Multiparameter estimation of Gaussian states}%
For a quantum system prepared in a state $\rho_\parVec$ parametrised by $\parVec\coloneqq \{\theta_j\}_j$ and a measurement with outcome $\xi$ distributed according to $p(\xi | \parVec)$, the \emph{classical} and \emph{quantum Fisher information matrices} (CFIM and QFIM) read, respectively~\cite{Rafal2020}:
\begin{align}
	\F_{jk} &\coloneqq \meanP{p(\xi | \parVec)}{\partial_{j}\!\ln p(\xi | \parVec)\,\partial_{k}\!\ln p(\xi | \parVec)},
	\label{eq:FCmatrix} \\
	\FQ_{jk} &\coloneqq \Tr[\rho_\parVec\,\tfrac{1}{2}\{L_j,L_k\}],
	\label{eq:FQmatrix}
\end{align}
where $\partial_j \equiv \partial/\partial_{\theta_j}$ is the derivative w.r.t.~any estimated parameter $\theta_j$, while by $\meanP{p(\xi | \parVec)}{\bullet}\coloneqq\int\!\d \xi\, p(\xi | \parVec) \bullet$ we denote the expected value. In the quantum case \eref{eq:FQmatrix}, $\meanP{p(\xi | \parVec)}{\bullet}$ naturally generalises to $\Tr[\rho_\parVec\bullet]$, while $\partial_{j}\!\ln p(\xi | \parVec)$ becomes the \emph{symmetric} logarithmic derivative $L_j$, defined as the solution to $\partial_{j}\rho_\parVec=\tfrac{1}{2}\{\rho_\parVec,L_j\}$~\cite{Helstrom1976,Holevo1982,BraunsteinCaves1994}. 

For any unbiased estimator $\tilde{\parVec}(\bm \xi)$ constructed based on the measurement data $\bm{\xi}=\{\xi_r\}_{r=1}^\nu$ gathered over $\nu$ independent shots,
its \emph{(squared-)error matrix} $\errormat\coloneqq\mean{(\tilde{\parVec}-\parVec)(\tilde{\parVec}-\parVec)^\T}$ satisfies the so-called \emph{quantum Cram\'{e}r-Rao bound} (QCRB)~\cite{Helstrom1976,Holevo1982}:   
\begin{equation}
	\nu\,\errormat \ge   \F^{-1} \ge \FQ^{-1}, 
	\label{eq:QCRB}
\end{equation}
where the first matrix inequality is guaranteed to be saturable by some $\tilde{\parVec}$ in the $\nu\to\infty$ limit, i.e.~for any $\pmat{W}\ge0$ for which $\Tr[\pmat{W}\errormat]$ is then minimised. In contrast, although the second inequality applies to any quantum measurement, as the optimal measurements for distinct parameters $\theta_j$ may not commute (formally $\Tr[\rho_\parVec[L_j,L_k]]\neq0$~\cite{Ragy2016}), it may \emph{not} be generally saturable by any estimator $\tilde{\parVec}$ given some $\pmat{W}\ge0$. However, it can differ at most by a factor of 2 from the minimal $\Tr[\pmat{W}\F^{-1}]$ attained by the optimal measurements~\cite{Rafal2020}. 

In the special case when only a single parameter, say $\theta_i$, is of interest, while the others are treated as nuisance ones (i.e.~$\pmat{W}_{jk}=\delta_{ij}\delta_{ik}$), we denote the lower bounds on the error in \eqnref{eq:QCRB} as $\cerror\theta_{i} \coloneqq \sqrt{[\F^{-1}]_{ii}}$ and $\qerror \theta_{i} \coloneqq \sqrt{[\FQ^{-1}]_{ii}}$, respectively, so that any unbiased estimator of $\theta_i$ satisfies then $\nu\MSE{\tilde{\theta}_i}\ge\cerror^2\theta_{i}\ge\qerror^2\theta_{i}$~\footnote{In fact, in the special case of $\pmat{W}_{jk}=\delta_{ij}\delta_{ik}$ an $\theta_i$-estimator and optimal measurement are guaranteed to exist such that $\nu\MSE{\tilde{\theta}_i}=\cerror^2\theta_{i}=\qerror^2\theta_{i}$~\cite{wojtek}}. This contrasts the \emph{ideal} single-parameter scenario with all parameters apart from $\theta_i$ known, in which case \eqnref{eq:QCRB} simplifies to $\nu\MSE{\tilde{\theta}_i}\ge\cerrorS^2\theta_{i}\ge\qerrorS^2\theta_{i}$ with $\cerrorS\theta_{i} \coloneqq 1/\sqrt{[\F]_{ii}}$ and $\qerrorS\theta_{i} \coloneqq 1/\sqrt{[\FQ]_{ii}}$.

Considering any Gaussian measurement of the probe outputs $\hat{A}_{\ell,\mrm{out}}$, its outcome $\svec{x}$ is normally distributed $\svec{x}\sim\exp\{-\tfrac{1}{2}(\svec{x}-\bar{\svec{x}})\smat{C}^{-1}(\svec{x}-\bar{\svec{x}})^T\}$~\cite{Cenni2022,Weedbrook2012}, with both the mean vector $\bar{\svec{x}}(\parVec)$ and the covariance matrix $\smat{C}(\parVec)$ depending on the parameter set $\parVec$. The CFIM \eref{eq:FCmatrix} takes then a special form~\cite{KayBook}:
\begin{equation}
	\F_{jk} = \frac{1}{2} \Tr\!\left[\smat{C}^{-1} (\partial_j \smat{C}) \smat{C}^{-1} (\partial_k \smat{C})  \right] + (\partial_j \bar{\svec{x}})^\T \smat{C}^{-1}  (\partial_k \bar{\svec{x}}),
	\label{eq:CFIM_Gaussian}
\end{equation}
so that in case of a heterodyne measurement being performed one should replace $\bar{\svec{x}} = \svec{S}_\mrm{out}^A$ and $\smat{C} = \smat{V}_\mrm{out}^A + \Id$ in the above~\cite{Weedbrook2012}. More generally, allowing for arbitrary quantum measurements performed on a Gaussian state of mean $\svec{S}(\parVec)$ and covariance $\smat{V}(\parVec)$, the QFIM \eref{eq:FQmatrix} reads $\FQ_{jk} = \frac{1}{2} \Tr\!\left[ \smat{L}_j \partial_k \smat{V} \right] + (\partial_j  \svec{S})^\T \smat{V}^{-1}  (\partial_k \svec{S})$~\cite{Monras2010,Monras2013,Nichols,Jing,Safranek2018}, with the matrix $\smat{L}_j$ possessing a nontrivial form detailed in \appref{app:NoisyApprox}. However, by generalising the results of~\citeref{Jiang2014} to the multiparameter case, we show in \appref{app:NoisyApprox} that QFIM can always be approximated for noisy Gaussian states, e.g.~highly thermalised, as
\begin{align}
	\FQ_{jk} &\approx  \frac{1}{2} \Tr\!\left[ \smat{V}^{-1} (\partial_{j} \smat{V}) \smat{V}^{-1} (\partial_k \smat{V})\right] + (\partial_j \svec{S})^\T \smat{V}^{-1}  (\partial_k  \svec{S}), 
	\label{eq:FmatrixNoisyparVec}
\end{align}
as long as the spectrum of $\smat{V}$ satisfies $\lambda_\mrm{min}(\smat{V})\gg1$~%
\footnote{\eqnref{eq:FmatrixNoisyparVec} is equivalently the leading term of the series-like expression for QFIM \eref{eq:FQmatrix} derived in~\cite{Safranek2018}}. Crucially, the \emph{noisy QFIM} \eref{eq:FmatrixNoisyparVec} is then determined by the response function \eref{eq:Gomega}, $\smat{G}_{\parVec}$, with $\svec{S}(\parVec) = \left(\Id - \kappa\smat{G}_{\parVec} \right) \svec{S}_\mrm{in}^A$ and $\smat{V}(\parVec)  = \left(\Id - \kappa\smat{G}_{\parVec} \right) \smat{V}_\mrm{in}^A \left(\Id - \kappa\smat{G}_{\parVec} \right)^\T + \kappa \smat{G}_{\parVec} \Xi \tilde{\smat{V}}_\mrm{in}^{B} \Xi^\T \smat{G}_{\parVec}^\T$.

\PRLsection{Single-parameter sensitivities}%
\label{sec:single_par}
When sensing a single parameter $\theta_0$ with others perfectly known, we set $\bar{\parVec} =0$ in \eqnref{eq:Gomega}, so that $\smat{H}_{\bar{\parVec} =0} =  \smat{H}$ and only the entry $j=k=0$ in Eq.~(\ref{eq:FmatrixNoisyparVec}) is relevant. Now, whenever the generator $\smat{H}$ is \emph{non-singular}, the response function \eref{eq:Gomega} admits a Neumann series $\smat{G}_{\theta_{0}} = -\sform \smat{H}^{-1} (\Id + \sum_{k=1}^{\infty} \theta_{0}^k (\smat{n}_{0} \smat{H}^{-1})^{k})$~%
\footnote{Not to be confused with a Taylor expansion in $\theta_0^{-1}$~\cite{LJ}, which generally cannot be trusted, see \appref{app:comp_to_LJ}.}
with $\lim_{\theta_0\to0}\smat{G}_{\theta_0} = -\sform \smat{H}^{-1}$. As a result, \eqnref{eq:FmatrixNoisyparVec} then reads $\FQ_{00}\approx C+\order{\theta_0}$ with some $\theta_0$-independent finite constant $C$, see \appref{app:single_par_nonsing}, so that the estimation error cannot vanish as $\theta_0\to0$, at which its minimal value is given by $\qerrorS\theta_{0} = 1/\sqrt{C}>0$~%
\footnote{The minimal error may still diminish, e.g., with the mean number of photons, $\bar{n}$, in the probes, for which $\svec{S}_\mrm{in},\smat{V}_\mrm{in}\propto \bar{n}$ and, hence, \eqnref{eq:FmatrixNoisyparVec} yields the \emph{standard quantum limit}~\cite{AC1}, i.e.~$\FQ_{00}\propto \bar{n}$ for $\bar{n}\gg1$. 
}.
In contrast, when the generator $\smat{H}$ is \emph{singular}, the Sain-Massey (SM) expansion~\cite{Sain1969,Howlett1982,Howlett2001,avrachenkov2013analytic} of the response function \eref{eq:Gomega} applies, i.e.~$\smat{G}_{\theta_0}  = \sform \theta_0^{-s} \sum_{k=0}^r \theta_0^k \X_k$ with each non-zero coefficient, $\X_k(\smat{H},\smat{n}_0)$, generally depending on $\smat{H}$ and the perturbation matrix $\smat{n}_0$ up to some (possibly infinite) $r$. The singularity is then characterized by the order $s\in \mathbb{N}_{+}$ of the pole in the SM expansion, which for any $\smat{G}_{\theta_0}$ is found by a recursive procedure described in \appref{app:SM}. \eqnref{eq:FmatrixNoisyparVec} implies then that $\FQ_{00} \approx \theta_0^{-2s} \left[A + \order{\theta_0}\right]$ for some $\theta_0$-independent $A>0$, see \appref{app:single_par_sing}, so that the estimation error $\qerrorS\theta_{0}  = 1/\sqrt{\FQ_{00}} \propto \theta_{0}^{s}$ vanishes as $\theta_0\to0$ at a rate dictated by $s$. The above observations prove that the singularity of $\smat{H}$ is essential for unbounded sensitivity, while other system properties, e.g.~exhibiting an EP and/or balancing the loss and gain rates~\cite{LJ}, can only play a role in determining the pole-order $s$. For the system considered, we evaluate in \appref{app:SM_examples} the SM expansions for the choices of the perturbation matrix $\smat{n}_0$ listed below \eqnref{eq:sens_H}. We observe that only in the case of two-mode symmetric frequency perturbations, $\smat{n}_0=\Id$~\cite{LJ}, the pole is second-order, while in all other cases it is of order one. We focus on the former in \figref{fig:CRB}, where we plot the exact, numerically obtained, estimation error (black lines)---neither noisy \eref{eq:FmatrixNoisyparVec} nor $\theta_0\approx0$ approximations are made---that consistently follows a quadratic scaling.

\begin{figure}[t!]
\includegraphics[width=\columnwidth]{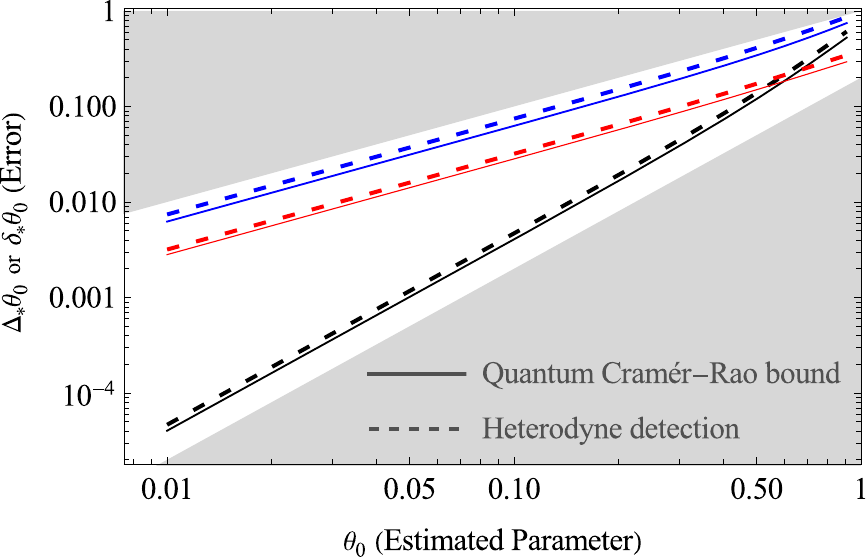}	
	\caption{\textbf{Estimation errors} attained by heterodyne detection (dashed lines), $\cerrorS\theta_{0}$ or $\cerror\theta_{0}$, vs quantum Cram\'{e}r-Rao bounds applicable to any measurement (solid lines), $\qerrorS\theta_{0}$ or $\qerror\theta_{0}$. Within the ideal single-parameter scenario, both $\errorS\theta_{0}$ (black) follow quadratic scaling. When $\theta_1$ constitutes a nuisance parameter whose variations maintain the singularity, the estimation errors $\errorM\theta_{0}$ scale linearly with $\theta_0$, e.g.~for $\theta_{1} = 0$ (red) or $\theta_{1} = 0.25$ (blue). In the above, all probe and scattering input modes are chosen to be in a thermal state with average photon-number $n_\mrm{A}=n_\mrm{B}=1$, while the intermode couplings are set to $\kappa=g=1$ and $\eta_1=1$, $\eta_2=3$.
	}
	\label{fig:CRB}
\end{figure}

\PRLsection{Impact of nuisance parameters}%
Dealing with the multiparameter scenario, in order to expand the corresponding response function \eref{eq:Gomega} around $\theta_0=0$, it is now $\smat{H}_{\bar{\parVec}}$ appearing in \eqnref{eq:Gomega} whose singularity must be verified. Still, when $\smat{H}_{\bar{\parVec}}$ is non-singular, $\smat{G}_{\parVec}$ admits again a Neumann series with $\lim_{\theta_0\to0}\smat{G}_{\parVec} = -\sform \smat{H}_{\bar{\parVec}}^{-1}$. Hence, $\FQ_{00}$ must be bounded as before, with the estimation error $\qerror\theta_{0}=\sqrt{[\FQ^{-1}]_{00}}\ge 1/ \FQ_{00}$ forbidden to vanish as $\theta_0\to0$. When $\smat{H}_{\bar{\parVec}}$ is singular instead, it is the SM expansion of $\smat{G}_{\parVec}$ that is valid, with the expansion coefficients $\X_k$ depending now also on nuisance parameters $\bar{\parVec}$. However, the sensitivity scaling in $\theta_{0}$ may not be associated directly with the pole-order any more, as the estimation error $\qerror\theta_{0}=\sqrt{[\FQ^{-1}]_{00}}$ involves the inverse of the QFIM and, hence, is generally affected by correlations between different unknown parameters. We show this explicitly by focussing on the two-parameter estimation scenario with the primary $\theta_0$ generated by $\smat{n}_0=\Id$, while $\theta_1$ acts as nuisance parameter of either $\HNS$ or $\HS$ in \eqnref{eq:HS_HNS}---black or red arrows in \figref{fig:par_space}\textbf{a}.

In case of $\HNS$, the error in estimating $\theta_0$ may not vanish as $\theta_0\to0$ for any $\theta_1\ne0$, for which $\smat{H}_{\theta_1}$ is \emph{non-singular}. We show this explicitly in \appref{app:nonsing_pert} by considering a thermal state with zero displacement as the probe input, for which $[\FQ^{-1}]_{00}= 1/\FQ_{00}\propto\theta_1^2$ at $\theta_0=0$. At the  singular point $\theta_1=0$, however, we also show---by resorting correctly to the SM expansion---that $\qerror\theta_{0}=\qerrorS\theta_{0}+O(\theta_0^4)$, with the impact of the nuisance parameter being then ignorable and single-parameter results being applicable ($\qerrorS\theta_{0}\propto\theta_0^2$, black lines in \figref{fig:CRB}). Turning to $\HS$, which importantly leads to $\smat{H}_{\theta_1}$ being \emph{singular} for any $\theta_1$, we demonstrate in \appref{app:sing_pert} that the SM expansion $\smat{G}_{\theta_0, \theta_1}= \sform \theta_0^{-2} \sum_{k=0}^1 \theta_0^k\,\X_k(\theta_1)$ exhibits a second-order pole. Substituting the expansion into \eqnref{eq:FmatrixNoisyparVec}, we obtain the entries of QFIM:~$\FQ_{00}\approx \alpha \theta_0^{-4}    +  2\beta \theta_0^{-3}  + \gamma \theta_0^{-2}$, $\FQ_{11}\approx  \alpha \theta_0^{-2}$ and $\FQ_{01}\approx \alpha \theta_0^{-3}+\beta \theta_0^{-2}$, with $\alpha,\gamma>0$ and $\beta\in\mathbb{R}$ (see \appref{app:sing_pert} for details). This implies that the error in estimating the primary parameter reads
\begin{align}
     \qerror\theta_{0} =  
     \left(\FQ_{00} - \frac{ \FQ_{01}  \FQ_{10}}{ \FQ_{11}}\right)^{-\frac{1}{2}}
     \!\approx
     \left(\gamma-\frac{\beta^2}{\alpha}\right)^{-\frac{1}{2}}\,\theta_0,
     \label{eq:error_nuis}
\end{align}
and scales linearly with $\theta_0$---the presence of nuisance $\theta_1$ precludes the scaling from following the quadratic behaviour dictated by the pole. We show this explicitly in \figref{fig:CRB}, where we plot the exact estimation error $\qerror\theta_{0}$ in red (blue) for $\HS$ with $\theta_1=0$ ($\theta_1=0.25$), as well as $\cerror\theta_{0}$ attained by heterodyne detection (dashed lines) that also follow the linear scaling.

\PRLsection{Conclusions}%
We establish the tools necessary to assess quantum Gaussian systems in sensing linear perturbations away from singularities. We investigate the divergence of sensitivity then exhibited, while clarifying that other dynamical properties, e.g.~operation at an exceptional point, fulfilment of lasing conditions or non-reciprocity, do not play a primary role. However, we demonstrate that nuisance parameters may then strongly affect the performance, i.e.~the rate at which the sensitivity diverges. Such a phenomenon resembles the setting of quantum superresolution problems~\cite{Tsang2016}, in which the lack of a spatial reference disallows to resolve infinitesimal separations between objects~\cite{Chrostowski2017,Rehacek2017,Grace2020,Sorelli2021,Almeida2021}. We leave open the question of how our results change, if one accounts for any prior knowledge about the sensed and/or nuisance parameters~\cite{Rubio2020,Rubio2021,Mehboudi2022}.

\begin{acknowledgments}
\PRLsection{Acknowledgments}%
We thank Dayou Yang, M{\u{a}}d{\u{a}}lin Gu{\c{t}}{\u{a}}, Marcin Jarzyna, Mohammad Mehboudi, Giacomo Sorelli and Konrad Banaszek for helpful comments. This work has been supported by the Quantum Optical Technologies project that is carried out within the International Research Agendas programme of the Foundation for Polish Science co-financed by the European Union under the European Regional Development Fund. RWC also acknowledges support from the Polish National Science Centre (NCN) under the Maestro Grant No.~DEC-2019/34/A/ST2/00081.
\end{acknowledgments}

\bibliographystyle{apsrev4-2}
\bibliography{EP_metro}


\appendix
\setcounter{secnumdepth}{3} 

\section{Linear quantum Langevin equations}
\label{app:Langevin_eqs}
Treating the optical modes at different frequencies $\omega'$ within the dissipation/driving and probe channels ($\hat{B}_{\ell}$ and $\hat{A}_{\ell}$ with $\ell=1,2$, respectively) as the ones forming an effective bath of the system, we can describe the complete dynamics by specifying apart from the system Hamiltonian, $\hat{H}_{S}$, the ones of the bath, $\hat{H}_B$, as well as the system-bath interaction, $\hat{H}_{SB}$, in the continuum limit~\cite{Gardiner1985,Steck}:
\begin{widetext}
	\begin{subequations}
		\label{continuum}
		\begin{align}
			\hat{H}_S    
			&= 
			\omega_1 \hat{a}_1^\dagger  \hat{a}_1 +  \omega_2 \hat{a}_2^\dagger  \hat{a}_2  + g \left(\hat{a}_1^\dagger \hat{a}_2 + {\rm H.c.} \right), \label{eq:cavityHam}\\
			\hat{H}_{B}  
			&=  
			\int_{-\infty}^{\infty} \!\d \omega^\prime\,\omega^\prime
			\left(   
			\hat{A}_{1,\omega^\prime}^\dagger \hat{A}_{1,\omega^\prime} +  
			\hat{A}_{2,\omega^\prime}^\dagger \hat{A}_{2,\omega^\prime} +  
			\hat{B}_{1,\omega^\prime}^\dagger \hat{B}_{1,\omega^\prime} +   
			\hat{B}_{2,\omega^\prime}^\dagger \hat{B}_{2,\omega^\prime}
			\right)\\
			\hat{H}_{SB} 
			&=  
			\int_{-\infty}^{\infty} \!\d \omega^\prime\,
			\left(
			\sqrt{\frac{\kappa(\omega^\prime)}{2\pi}} \left[\hat{a}_1 \hat{A}_{1,\omega^\prime}^\dagger + {\rm H.c.} \right] + 
			\sqrt{\frac{\kappa(\omega^\prime)}{2\pi}} \left[\hat{a}_2 \hat{A}_{2,\omega^\prime}^\dagger + {\rm H.c.} \right] \right. \nonumber \\
			&\qquad\qquad\qquad
			\left.+ \sqrt{\frac{\eta_1(\omega^\prime)}{2\pi}} \left[\hat{a}_1 \hat{B}_{1,\omega^\prime}^\dagger + {\rm H.c.} \right] +  
			\sqrt{\frac{\eta_2(\omega^\prime)}{2\pi}} \left[ \hat{a}_2^\dagger  \hat{B}_{2,\omega^\prime}^\dagger  + {\rm H.c.}  \right]
			\right),
			\label{eq:H_SB}
		\end{align}
	\end{subequations}
\end{widetext}	                 
with each optical mode satisfying for any $\omega^\prime$, $\omega^\dprime$:
\begin{equation}
	\left[\hat{A}_{\ell,\omega^\prime}, \hat{A}_{\ell,\omega^\dprime}^\dagger\right]=
	\left[\hat{B}_{\ell,\omega^\prime}, \hat{B}_{\ell,\omega^\dprime}^\dagger\right] = 
	\delta(\omega^\prime  - \omega^\dprime).
	\label{eq:comm_rel}
\end{equation}

As a result, defining the total Hamiltonian as $\hat{H}_T:=\hat{H}_S+\hat{H}_B+\hat{H}_{SB}$, we may write the Heisenberg-Langevin equations of motion describing dynamics of cavity modes, $\hat{a}_1(t)$ and $\hat{a}_2(t)$, in the Heisenberg picture as ($\hbar=1$):
\begin{align}
	\label{eq:a12dot}
	\partial_t \hat{a}_{\ell}(t)
	&= -\i \left[\hat{a}_{\ell}(t), \hat{H}_T \right] = -\i \left[\hat{a}_{\ell}(t), \hat{H}_S+\hat{H}_{SB} \right], \nonumber \\
	&= -\i \omega_1 \,\hat{a}_{\ell}(t) - \i g \,\hat{a}_{\neg{l}}(t)  - \i\!\int_{-\infty}^{\infty} \!\!\!\d \omega^\prime \sqrt{\frac{\kappa(\omega^\prime)}{2\pi}} \hat{A}_{\ell,\omega^\prime}(t) \nonumber   \\
	&- \i  \int_{-\infty}^{\infty} \!\d \omega^\prime\, \sqrt{\frac{\eta_{\ell}(\omega^\prime)}{2\pi}} \hat{O}_{\ell,\omega^\prime}(t),
\end{align}
where $\hat{O}_2:=\hat{B}^\dagger_2$ in contrast to $\hat{O}_1:=\hat{B}_1$---with the second cavity being actually driven via the channel $\hat{B}_2$, see \eqnref{eq:H_SB}.

In analogous way, we obtain the Heisenberg equations for all the optical modes as
\begin{align}
	\label{eq:A_modes_langevin}
	\partial_t \hat{A}_{\ell,\omega^\prime}(t)
	& = - \i \omega^\prime \hat{A}_{\ell,\omega^\prime}(t)-\i \sqrt{\frac{\kappa(\omega^\prime)}{2\pi}}\hat{a}_{\ell}(t) \\
	\label{eq:B_modes_langevin}
	\partial_t \hat{B}_{\ell,\omega^\prime}(t)
	& = - \i \omega^\prime \hat{B}_{\ell,\omega^\prime}(t)-\i \sqrt{\frac{\eta_{\ell}(\omega^\prime)}{2\pi}}\hat{o}_{\ell}(t)
\end{align}
with $\hat{o}_1:=\hat{a}_1$ and $\hat{o}_2:=\hat{a}^\dagger_2$ according to \eqnref{eq:H_SB}. Hence, starting from some remote initial time $t_i$ at which the operators $\hat{A}_{\ell,\omega^\prime}(t_i)$ and $\hat{B}_{\ell,\omega^\prime}(t_i)$ satisfy the same commutation relations as in \eqnref{eq:comm_rel}, the equations of motion (\ref{eq:A_modes_langevin}-\ref{eq:B_modes_langevin}) can be formally integrated, yielding
\begin{align}
	\label{eq:A_modes_t}
	\hat{A}_{\ell,\omega^\prime}(t)
	& =
	\hat{A}_{\ell,\omega^\prime}(t_i)\,\e^{-\i\omega^\prime(t-t_i)} \nonumber \\
	& - \i\sqrt{\frac{\kappa(\omega^\prime)}{2\pi}} \int_{t_i}^t \!\d D\, \hat{a}_{\ell}(D)\,\e^{-\i\omega^\prime(t-D)}
\end{align}
and similarly for $\hat{B}_{\ell,\omega^\prime}(t)$.

Now, substituting the above integral solutions into \eqnref{eq:a12dot} and invoking the \emph{Markov approximation} under which the interaction rates are assumed to be frequency independent, i.e.:
\begin{equation}
	\kappa (\omega^\prime) \approx \kappa, \quad \eta_1(\omega^\prime) \approx \eta_1, \quad \eta_2(\omega^\prime) \approx \eta_2,
\end{equation} 
we may rewrite \eqnref{eq:a12dot} in a vector-form consistent with \eqnref{eq:dyn_model} (dropping w.l.o.g.~the $t$-dependence):
\begin{align}\label{eq:a1a2}
	\begin{pmatrix}
		\partial_t \hat{a}_1  \\  \partial_t  \hat{a}_2 
	\end{pmatrix}                    
	&=  -\i  
	\begin{pmatrix}
		\omega_0 - \i (\eta_1 + \kappa)/2     &        g  \\
		g                                &            \omega_0 + \i (\eta_2 - \kappa)/2              
	\end{pmatrix}
	\begin{pmatrix}
		\hat{a}_1   \\   \hat{a}_2
	\end{pmatrix}  \nonumber \\&+   \begin{pmatrix}    \sqrt{\kappa}  \hat{A}_{1,\mrm{in}}  \\ \sqrt{\kappa}  \hat{A}_{2,\mrm{in}}        \end{pmatrix} 
	+ \begin{pmatrix}
		\sqrt{\eta_1} \hat{B}_{1,\mrm{in}}        \\   - \sqrt{\eta_2} \hat{B}_{2,\mrm{in}}^\dagger 
	\end{pmatrix},
\end{align}
where for each optical mode, $\hat{\mathcal{O}}_{\omega^\prime} \in \{\hat{A}_{\ell,\omega^\prime},\hat{B}_{\ell.\omega^\prime}\}$ with $\ell=1,2$, we define its effective input field as
\begin{equation}
	\hat{\mathcal{O}}_{\mrm{in}}(t) := -\i \int_{-\infty}^{\infty} \! \frac{\d\omega^{\prime}}{\sqrt{2\pi}}\, \hat{\mathcal{O}}_{\omega^\prime}(t_i)\, \e^{-\i \omega^\prime (t - t_i)},
\end{equation}
which satisfies $\left[\hat{\mathcal{O}}_{\mrm{in}}(t),\hat{\mathcal{O}}_{\mrm{in}}(t')^\dagger\right]=\delta(t-t')$.

On the other hand, we may define the measured output fields of the monitored modes ($\hat{A}_\ell$ with $\ell=1,2$) as 
\begin{equation}
	\hat{A}_{\ell,\mrm{out}}(t) := -\i \int_{-\infty}^{\infty} \!\frac{\d\omega^{\prime}}{\sqrt{2\pi}}\, \hat{A}_{\ell,\omega^\prime}(t_f)\, \e^{-\i \omega^\prime (t - t_f)},
\end{equation}
where $t_f$ is some final time in the remote future, so that by integrating \eqnref{eq:A_modes_langevin} now backwards in time we have
\begin{align}
	\label{eq:A_modes_t_future}
	\hat{A}_{\ell,\omega^\prime}(t)
	& =
	\hat{A}_{\ell,\omega^\prime}(t_f)\,\e^{-\i\omega^\prime(t-t_f)} \nonumber \\
	& + \i\sqrt{\frac{\kappa(\omega^\prime)}{2\pi}} \int_t^{t_f} \!\d D\, \hat{a}_{\ell}(D)\,\e^{-\i\omega^\prime(t-D)}.
\end{align}
Hence, by equating \eqnsref{eq:A_modes_t}{eq:A_modes_t_future} and performing the integral over frequencies $\omega'$, while invoking again the Markov approximation, we obtain the \emph{input-output relation} for the probe fields as:
\begin{equation}\label{eq:IO}
	\hat{A}_{\ell,\mrm{out}}(t)  = \hat{A}_{\ell,\mrm{in}}(t)  - \sqrt{\kappa}\, \hat{a}_{\ell}(t).
\end{equation}

\section{Gaussian dynamics in the Fourier domain}
\label{app:dyns_omega}
Let us define the Fourier transform for any vector of time-dependent bosonic operators, say $\qvec{c}(t)=\{\hat{c}_1(t),\hat{c}_2(t),\dots\}^\T$, as $\qvec{c}[\omega]:=\mathcal{F}_{\omega} [\qvec{c}] := \int\!dt\, \e^{\i \omega t} \qvec{c}(t)$, so that the vectors containing then the corresponding position and momentum quadratures of all the modes in the Fourier basis read:~$\qvec{q}^c[\omega]=\qvec{c}[\omega] + \left(\qvec{c}[\omega]\right)^\dagger$ and $\qvec{p}^c[\omega]=-\i (\qvec{c}[\omega] - \left(\qvec{c}[\omega]\right)^\dagger)$, respectively, with the Hermitian conjugation being performed entry-wise, i.e.~$\left(\qvec{c}[\omega]\right)^\dagger=\{\hat{c}_1[\omega]^\dagger,\hat{c}_2[\omega]^\dagger,\dots\}^\T$. However, let us observe that $\left(\qvec{c}[\omega]\right)^\dagger = \left(\int\!dt\, \e^{\i \omega t} \qvec{c}(t)\right)^\dagger = \int\!dt\, \e^{-\i \omega t} \qvec{c}^\dagger(t)=\mathcal{F}_{-\omega}[\qvec{c}^\dagger]=\qvec{c}^\dagger[-\omega]$, so that the conjugate vector can be equivalently interpreted as containing Fourier transforms of operator conjugates, but evaluated at the \emph{negative} frequency.

As a result, we may rewrite the dynamical equation \eqref{eq:a1a2} in the Fourier domain as 
\begin{equation}
	- \i \omega \qvec{a}[\omega] = -\i (\omega_0 \id + \H) \qvec{a}[\omega] + \qvec{A}_\mrm{in}[\omega] + 	\qvec{B}_\mrm{in}[\omega],
	\label{eq:dyn_fourier}
\end{equation}
where we have defined similarly to \eqnref{eq:dyn_model}:~$\qvec{a}:=\{\hat{a}_1,\hat{a}_2\}^\T$, $\qvec{A}_\mrm{in}:=\{\sqrt{\kappa} \hat{A}_{1,\mrm{in}}, \sqrt{\kappa}\hat{A}_{2,\mrm{in}}\}^\T$ and $\qvec{B}_\mrm{in}:=\{\sqrt{\eta_1} \hat{B}_{1,\mrm{in}},- \sqrt{\eta_2} \hat{B}_{2,\mrm{in}}^\dagger\}^\T$ as the operator-vectors containing cavity, probing and scattering modes, respectively. Hence, by re-expressing any (creation or annihilation) operator appearing in \eqnref{eq:dyn_fourier} in terms of its quadratures, i.e.~$\hat{o}[\omega]=\tfrac{1}{2}(\hat{q}^o[\omega]+\i\hat{p}^o[\omega])$, and using the fact that $\hat{q}^{o^\dagger}[\omega]=\hat{q}^{o}[-\omega]$ and $\hat{p}^{o^\dagger}[\omega]=-\hat{p}^{o}[-\omega]$,
we may rewrite \eqnref{eq:dyn_fourier} with help of the quadrature vectors:
\begin{align}
	\qsvec{S}^{S}[\omega] &:= 
	\begin{pmatrix}
		\hat{q}_1 [\omega] \\ \hat{q}_2[\omega]  \\ \hat{p}_1 [\omega]  \\  \hat{p}_2[\omega]
	\end{pmatrix}, \quad
	\qsvec{S}_\mrm{in}^{\bullet}[\omega] := 
	\begin{pmatrix}
		\hat{q}^{\bullet_\mrm{in}}_1[\omega] \\ \hat{q}^{\bullet_\mrm{in}}_2[\omega]  \\ \hat{p}^{\bullet_\mrm{in}}_1[\omega]  \\  \hat{p}^{\bullet_\mrm{in}}_2[\omega]
	\end{pmatrix}  
	\label{eq:Svectors}
\end{align}
with $\bullet=A/B$ representing either probing or scattering modes, in order to importantly obtain the phase-space representation of the dynamics \eref{eq:a1a2} in the form:
\begin{align}
	-\omega\sform \qsvec{S}^{S}[\omega]  
	& = 
	-(\omega_{0} \Id + \smat{H} ) \sform \qsvec{S}^{S}[\omega] + \sqrt{\kappa} \qsvec{S}^{A}_\mrm{in}[\omega] 
	\nonumber \\
	&+  \smat{K}_{B_1}\qsvec{S}^{B}_\mrm{in}[\omega] + \smat{K}_{B_2}\qsvec{S}^{B}_\mrm{in}[-\omega],
	\label{eq:dyn_phasespace}
\end{align}
where any (2x2) matrix $\mat{m}$ appearing in \eqnref{eq:dyn_fourier} is now rewritten to act in the (symplectic) phase space as
\begin{equation}
	\smat{m}
	\equiv	\stransf{\mat{m}} 
	:= 	\begin{pmatrix}
		\Re[\mat{m}]   & - \Im[\mat{m}] \\
		\Im[\mat{m}]  & \Re[\mat{m}]
	\end{pmatrix}, 	
	\label{eq:phase_space_rep}
\end{equation}
so that $\smat{H}=\stransf{\mat{H}}$ consistently with \eqnref{eq:H0_smatrix}, $\Id=\stransf{\id}$ is just a 4x4 identity matrix, and $\sform:=\stransf{\i\id}$ is the symplectic form playing the role of the imaginary unit appearing in \eqnref{eq:dyn_fourier}. Moreover, we have defined above $\smat{K}_{B_1}:=\mrm{diag}\{\sqrt{\eta_1},0,\sqrt{\eta_1},0\}$ and $\smat{K}_{B_2}:=\mrm{diag}\{0,-\sqrt{\eta_2},0,\sqrt{\eta_2}\}$ as the effective matrices coupling the system to the first and the second scattering channels, $\hat{B}_1$ and $\hat{B}_2$, respectively. As depicted in \figref{fig:model}, these are independently causing dissipation and gain in the cavities, so distinct modes at either positive or negative frequencies contribute in the two cases---and so the quadratures of $\hat{B}_{1,\mrm{in}}[\omega]$ and $\hat{B}_{2,\mrm{in}}[-\omega]$ appear in \eqnref{eq:dyn_phasespace}.

Performing the Fourier transform of the input-output relation \eref{eq:IO} in a similar manner, we obtain 
\begin{equation}
	\qsvec{S}_\mrm{out}^{A}[\omega] = \qsvec{S}_\mrm{in}^{A}[\omega ] - \sqrt{\kappa} \qsvec{S}_\mrm{in}^{S}[\omega]
	\label{eq:in-out_phasespace}
\end{equation}
with $\qsvec{S}_\mrm{out}^{A}$ being analogously defined to $\qsvec{S}_\mrm{in}^{A}$ in \eqnref{eq:Svectors} for the output probing modes, i.e.~the ones being actually measured. 

Hence, we can now eliminate the system degrees of freedom in \eqnref{eq:dyn_phasespace} with help of \eref{eq:in-out_phasespace}, in order to obtain the expression for the output quadratures being probed in terms of the input modes, i.e.
\begin{align}\label{eq:SAout}
	\qsvec{S}_\mrm{out}^{A}[\omega] &= (\Id  - \kappa \smat{G}[\omega]) \qsvec{S}_\mrm{in}^{A}[\omega] \nonumber \\& - \sqrt{\kappa} \smat{G}[\omega] \left( \smat{K}_{B_1} \qsvec{S}_\mrm{in}^{B}[\omega] + \smat{K}_{B_2}  \qsvec{S}^{B}_\mrm{in}[-\omega] \right),
\end{align}
where $\smat{G}[\omega] := \sform \left((\omega-\omega_0)\Id - \smat{H}\right)^{-1}$ is the response (Green's) function stated in \eqnref{eq:Greens}.

Denoting the average amplitude of $\qsvec{S}$ by $\svec{S} := \langle \qsvec{S} \rangle$, and further assuming the scattering channels to be initialised with $\svec{S}^{B}_\mrm{in}[\omega]   =  \svec{S}^{B}_\mrm{in}[-\omega]   = 0$, we conclude that 
\begin{equation}\label{eq:Sout}
	\svec{S}^{A}_\mrm{out}[\omega] =  \left(\Id  - \kappa \smat{G}[\omega]\right) \svec{S}_\mrm{in}^{A}[\omega].
\end{equation}

In order to find a similar input-output relation determining the covariance matrix of the output modes being measured, i.e.~$\smat{V}_\mrm{out}^A[\omega]$ with entries $[\smat{V}_\mrm{out}^A]_{jk}=\frac{1}{2} \qmean{ \{[\qsvec{S}_\mrm{out}^{A}]_j, [\qsvec{S}_\mrm{out}^{A}]_k\}}  - \qmean{[\qsvec{S}_\mrm{out}^{A}]_j}\qmean{[\qsvec{S}_\mrm{out}^{A}]_k}$, we formally define the scattering matrix $\mathcal{S}$ for the entire input-output relation connecting all the input and output modes as follows 
\begin{equation}
	\begin{pmatrix}
		\qsvec{S}^{A}_\mrm{out}[\omega] \\ \qsvec{S}^{B}_\mrm{out}[\omega] \\ \qsvec{S}^{B}_\mrm{out}[-\omega] 
	\end{pmatrix}  = \mathcal{S}[\omega] \begin{pmatrix}
		\qsvec{S}^{A}_\mrm{in}[\omega] \\ \qsvec{S}^{B}_\mrm{in}[\omega] \\ \qsvec{S}^{B}_\mrm{in}[-\omega]
	\end{pmatrix}.
\end{equation}
However, as we are interested only in the probe output $A$, we observe that it is sufficient to determine only the first row of the $\mathcal{S}$-matrix that is fully specified by the relation \eref{eq:SAout}, i.e.:  
\begin{equation}
	\mathcal{S}[\omega] =  
	\left(
	\begin{array}{c|c}
		\Id - \kappa \smat{G}[\omega] & -\sqrt{\kappa} \smat{G}[\omega] \Xi\\
		\hline
		\dots&\dots 
	\end{array}
	\right),
	\label{eq:Smatrix}
\end{equation}
where $\Xi := (\smat{K}_{B_1}~|~\smat{K}_{B_2})$ is a 4x8 matrix, responsible for coupling the scattering modes both at the positive and negative frequency to the measured output mode. As the probing modes are uncorrelated from the scattering ones at the input, we may generally write the total covariance matrix of all the output modes as $\smat{V}_\mrm{out}[\omega] = \mathcal{S} (\smat{V}_\mrm{in}^{A}[\omega]\oplus \tilde{\smat{V}}_\mrm{in}^{B}[\omega]) \mathcal{S}^\T$~\cite{Ferraro2005}, where by $\tilde{\smat{V}}_\mrm{in}^{B}[\omega]$ we effectively denote the overall (8x8) covariance matrix describing correlations between all input scattering modes, i.e.~also between ones at $+\omega$ and $-\omega$. 

As a consequence, focussing only on the output modes being measured, it follows from \eqnref{eq:Smatrix} that
\begin{align}\label{eq:Vout}
	\smat{V}_\mrm{out}^{A}[\omega] &= \left(\Id - \kappa \smat{G}[\omega]\right) \smat{V}_\mrm{in}^{A}[\omega] \left(\Id - \kappa \smat{G}[\omega]\right)^\T \nonumber\\ 
	&+ \kappa \smat{G}[\omega] \Xi \tilde{\smat{V}}_\mrm{in}^{B}[\omega] \Xi^\T \smat{G}[\omega]^\T,
\end{align}
where the second term can be further simplified by noting that $\tilde{\smat{V}}_\mrm{in}^{B}[\omega] = \smat{V}_\mrm{in}^B[\omega] \oplus \smat{V}_\mrm{in}^{B}[-\omega]$, which yields 
\begin{align}
	\Xi \tilde{\smat{V}}_\mrm{in}^{B}[\omega] \Xi^\T  &= \smat{K}_{B_1} \smat{V}_\mrm{in}^{B}[\omega] \smat{K}_{B_1}^\T + \smat{K}_{B_2} \smat{V}_\mrm{in}^{B}[-\omega] \smat{K}_{B_2}^\T.
\end{align}

\section{Quantum Fisher information matrix (QFIM) for Gaussian states and its noisy form}
\label{app:NoisyApprox}
In general, the $jk$-th element of the quantum Fisher information matrix (QFIM), corresponding to $\theta_j$ and $\theta_k$ parameters, is given for any $m$-mode Gaussian state with amplitude vector $\svec{S}$ and covariance matrix $\smat{V}$ by~\cite{Monras2010,Monras2013,Nichols,Jing}:
\begin{align}
	\FQ_{jk} &= \frac{1}{2} \Tr\!\left[ \smat{L}_j \partial_k \smat{V} \right] + 2 (\partial_j  \svec{S})^\T \smat{V}^{-1}  (\partial_k \svec{S}), 
	\label{eq:FqGaussian}
\end{align}
where the matrix 
\begin{equation}
	\smat{L}_j = \sum\limits_{a,b=1}^{m} \sum\limits_{l=0}^{3} \frac{g_{l;j}^{(ab)} }{ v_a v_b + (-1)^{l+1}} \left(\mathcal{D}^\T \right)^{-1} B_{l}^{(ab)} \mathcal{D}^{-1}
	\label{eq:L-matrix}
\end{equation} 
can be interpreted as the symmetric logarithmic derivative (w.r.t.~the parameter $\theta_j$) associated with the covariance matrix $\smat{V}$, which is diagonalised by the symplectic matrix $\mathcal{D}$ (i.e.~satisfying $\mathcal{D} \smat{J} \mathcal{D}^\T =\smat{J}$) such that $\smat{V}_d= \mathcal{D}^{-1} \smat{V}(\mathcal{D}^\T)^{-1} = \mrm{diag}\{v_1,v_2,\dots,v_m\} \otimes \Id_2$ is diagonal with symplectic eigenvalues $v_k$. The $g$-coefficients in \eqnref{eq:L-matrix} read
\begin{equation}
	g_{l;j}^{(ab)} = \Tr\!\left[ \mathcal{D}^{-1} (\partial_j \smat{V})  (\mathcal{D}^\T)^{-1} B_{l}^{(ab)}  \right],
\end{equation} 
where $\{B_{l}^{(ab)} = \cP A_{l}^{(ab)} \cP^\T\}_{l;ab}$ constitute a complete basis of matrices, with $\cP$ being a permutation matrix (defined below) and $A_{l}^{(ab)}$ forming a $2m$-dimensional matrix with only non-zero entries within the $2\times2$ block at the position $ab$, chosen according to the index $l=0,1,2,3$ from the ordered set $\frac{1}{\sqrt{2}} \left\{\i \sigma_y, \sigma_z, \id, \sigma_x \right\}$. 

For completeness, note that in \citeref{Nichols} the QFIM is derived in the basis $\qsvec{R}=\{\hat{q}_1,\hat{p}_1, \dots, \hat{q}_n,\hat{p}_n\}^\T$, with the corresponding covariance matrix $\sigma$ being thus diagonalised by a different symplectic matrix $\mathcal{K}$ such that $\sigma_d=\mathcal{K}^{-1} \sigma (\mathcal{K}^\T)^{-1}$. However, as we consider here simply a different ordering of the quadratures, i.e.~$\qsvec{S} = \cP \qsvec{R}$ with $\cP$ being a permutation matrix with entries $\cP_{kl} = \delta_{k, 2l-1}$ for $1 \le l \le n$ and $\delta_{k, 2l-2n}$ for $ n+1 \le l \le 2n$~\cite{Ferraro2005}, the symplectic matrix $\mathcal{D}$ appearing above is generally related to $\mathcal{K}$ via $\mathcal{D}=\cP\mathcal{K}\cP^\T$.

In order to arrive at the noisy approximation of the QFIM given by \eqnref{eq:FmatrixNoisyparVec}, we prove first the following 
\begin{lemma}
	The $\smat{L}$-matrix defined in \eqnref{eq:L-matrix} constitutes the solution to the matrix equation:
	\begin{align}\label{eq:LjImplicit}
		\partial_j  \smat{V}
		&= \smat{V} \smat{L}_j \smat{V} + \sform \smat{L}_j \sform.
	\end{align}
\end{lemma}	
\begin{proof}
	Introducing for short the notation $\overline{\smat{O}}:=\mathcal{D}^{-1} \smat{O} (\mathcal{D}^\T)^{-1}$ to denote any given matrix $\smat{O}$ written in the eigenbasis of $\smat{V}$ (for which $\overline{\smat{V}} = \smat{V}_d$), we rewrite \eqnref{eq:LjImplicit} as
	\begin{align}           	
		\overline{\partial_j \smat{V}} = \smat{V}_d \widetilde{\smat{L}}_j  \smat{V}_d +  \smat{J} \widetilde{\smat{L}}_j    \smat{J}, 
		\label{eq:Dtranf}
	\end{align}
	where we have also defined $\widetilde{\smat{L}}_j := \mathcal{D}^\T \smat{L}_j \mathcal{D} $. 
	Furthermore, defining the expansions of $\overline{\partial_j \smat{V}}$ and $\widetilde{\smat{L}}_j$ in basis of matrices $\{B_{l}^{(ab)}\}_{l;ab}$ specified above, i.e.:
	\begin{equation}
		\label{eq:gh}
		\overline{\partial_j \smat{V}} = \sum\limits_{abl} g_{l}^{(ab)} B_{l}^{(ab)}, 
		\qquad 
		\widetilde{\smat{L}}_j  = \sum\limits_{abl} h_{l}^{(ab)} B_{l}^{(ab)},
	\end{equation}
	we substitute these into \eqnref{eq:Dtranf}, in order to obtain 
	\begin{align}
		\sum\limits_{abl} g_{l}^{(ab)} B_{l}^{(ab)} &= \sum\limits_{abl} h_{l}^{(ab)}  \left(  \smat{V}_d B_{l}^{(ab)}  \smat{V}_d  + \smat{J} B_{l}^{(ab)} \smat{J}  \right), \nonumber \\
		&= \sum\limits_{abl} h_{l}^{(ab)} \left(  v_a v_b B_{l}^{(ab)} + (-1)^{l+1} B_{l}^{(ab)}  \right).
	\end{align}
	Hence, equating the coefficients in the equation above, we have that each
	\begin{equation}
		h_{l}^{(ab)} = \frac{g_{l}^{(ab)}}{v_a v_b + (-1)^{l+1}},
	\end{equation}
	so that it follows from \eqnref{eq:gh} that 
	\begin{equation}
		\widetilde{\smat{L}}_j =  \sum\limits_{abl} \frac{g_{l}^{(ab)}}{v_a v_b  + (-1)^{l+1}}  B_{l}^{(ab)},
	\end{equation}
	which, as required, yields then the expression for $\smat{L}_j=\left(\mathcal{D}^\T \right)^{-1} \widetilde{\smat{L}}_j \mathcal{D}^{-1}$ stated in \eqnref{eq:L-matrix}. 
\end{proof}

Now, let us derive the noisy approximation of the QFIM that applies to any noisy, e.g.~highly thermalised, Gaussian state, by first considering the $\smat{L}$-matrix \eref{eq:L-matrix}:

\begin{thm}[Noisy approximation \cite{Jiang2014}]
	If all the eigenvalues of the covariance matrix $\smat{V}$ fulfil $\forall_{\alpha=0}^{2m}:\!\,v_\alpha \gg 1$, then the $\smat{L}$-matrix defined in \eqnref{eq:L-matrix} w.r.t.~any parameter, say $\theta_j$, can be approximated as follows
	\begin{equation}
		\smat{L}_j	\approx \smat{V}^{-1} (\partial_j \smat{V})  \smat{V}^{-1}.
		\label{eq:SLD_noisy_approx}
	\end{equation} 
	\label{thm:noisy}
\end{thm}	
\begin{proof}
	Following \citeref{Jiang2014}, we sandwich \eqnref{eq:Dtranf} independently by the symplectic form $\sform$ and the diagonal form of $\smat{V}$, in order to respectively obtain
	\begin{align}
		\smat{J} \overline{\partial_j \smat{V} } \smat{J} & = \smat{J} \smat{V}_d \widetilde{\smat{L}}_j  \smat{V}_d \smat{J} +    \widetilde{\smat{L}}_j,  \\
		\smat{V}_d   \overline{\partial_j \smat{V} } \smat{V}_d &= \smat{V}_d^2 \widetilde{L}_j \smat{V}_d^2 + \smat{V}_d \smat{J} \widetilde{\smat{L}}_j \smat{J} \smat{V}_d.
	\end{align}
	Subtracting the latter from the former and using the fact that $\smat{J}$ commutes with $\smat{V}_d$, the above equations lead to
	\begin{equation}\label{eq:diff}
		\widetilde{\smat{L}}_j - \smat{V}_d^2 \widetilde{L}_j \smat{V}_d^2 =  \smat{J} \overline{\partial_j \smat{V} } \smat{J} - \smat{V}_d   \overline{\partial_j \smat{V} } \smat{V}_d.
	\end{equation}
	Moreover, specifying explicitly the spectral decomposition of the covariance matrix as $\smat{V} = \sum_{\alpha=1}^{2m} v_\alpha \svec{e}_\alpha \svec{e}_\alpha^\T$ with the eigenvalues $v_\alpha$ being doubly degenerate, we may rewrite \eqnref{eq:diff} in the eigenbasis of $\smat{V}$ as ($[\bullet]_{\alpha \beta}=\svec{e}^\T_\alpha\bullet\svec{e}_\beta$):
	\begin{equation}
		[\widetilde{\smat{L}}_j]_{\alpha \beta} \left(1 - v_\alpha^2 v_\beta^2 [\overline{\partial_j \smat{V} }]_{\alpha \beta} \right) = [\smat{J} \overline{\partial_j \smat{V} } \smat{J} ]_{\alpha \beta}  - v_{\alpha} v_{\beta}  [\overline{\partial_j \smat{V} } ]_{\alpha \beta}.
	\end{equation}
	Now, this allows us to write each such entry of $\widetilde{\smat{L}}_j$ as
	\begin{align}
		[\widetilde{\smat{L}}_j]_{\alpha \beta}
		&=  
		\frac{v_{\alpha} v_{\beta}  [\overline{\partial_j \smat{V} } ]_{\alpha \beta}-[\smat{J} \overline{\partial_j \smat{V} } \smat{J} ]_{\alpha \beta}}{v_\alpha^2 v_\beta^2 [\overline{\partial_j \smat{V} }]_{\alpha \beta} -1} 
		\label{eq:appr0}\\
		&\approx  \frac{1}{v_\alpha^2 v_\beta^2}\left( v_{\alpha} v_{\beta}  [\overline{\partial_j \smat{V} } ]_{\alpha \beta} - [\smat{J} \overline{\partial_j \smat{V} } \smat{J} ]_{\alpha \beta} \right)  
		\label{eq:appr1}\\
		&\approx  \frac{1}{v_\alpha v_\beta} [\overline{\partial_j \smat{V} } ]_{\alpha \beta},
		\label{eq:appr2}
	\end{align}
	where we have assumed that all the eigenvalues $\nu_{\alpha/\beta}$ are much larger that $1$, so that in the line \eref{eq:appr0} the denominator is assured to be non-zero, while the terms $1$ and $\smat{J} \overline{\partial_j \smat{V} } \smat{J}$ are then ignored in \eref{eq:appr1} and \eref{eq:appr2}, respectively.
	
	Hence, returning to the matrix form, this translates to
	\begin{align}
		\widetilde{\smat{L}}_j
		&= 
		\sum_{\alpha,\beta=1}^{2m} [\widetilde{\smat{L}}_j]_{\alpha\beta} \,\svec{e}_\alpha \svec{e}_\beta^\T \\
		&\approx 
		\sum_{\alpha,\beta=1}^{2m} \frac{1}{v_\alpha v_\beta} [\overline{\partial_j \smat{V} }]_{\alpha \beta} \,\svec{e}_\alpha \svec{e}_\beta^\T \\
		&=
		\smat{V}_d ^{-1} (\overline{\partial_j \smat{V}}) \smat{V}_d^{-1},
	\end{align}
	which under transformation to the original basis of $\smat{V}$, i.e.~$\smat{L}_j=\left(\mathcal{D}^\T \right)^{-1} \widetilde{\smat{L}}_j \mathcal{D}^{-1}$, yields \eqnref{eq:SLD_noisy_approx} as desired.
\end{proof}

Finally, using the approximation \eref{eq:SLD_noisy_approx} for the matrix $\smat{L}_j$, the QFIM \eref{eq:FqGaussian} simplifies to \eqnref{eq:FmatrixNoisyparVec}, i.e: 
\begin{align}
	\FQ_{jk} &\approx  \frac{1}{2} \Tr\!\left[ \smat{V}^{-1} \left(\partial_j \smat{V}\right) \smat{V}^{-1} \left( \partial_k \smat{V} \right) \right] +2 (\partial_j  \svec{S})^\T \smat{V}^{-1}  (\partial_k \svec{S}),
	\label{eq:FmatrixNoisyApp}
\end{align}
which we term, for short, the \emph{noisy QFIM}.

\section{Noisy QFIM for $\parVec$-parametrised linear perturbations}
\label{app:QFIM_perturb}
Let us consider the task of sensing linear perturbations of a dynamical generator $\H$, e.g.~the one appearing in \eqnref{eq:dyn_fourier}, or equivalently \eqnref{eq:dyn_model}, under which its form is modified as in \eqnref{eq:sens_H}, i.e.:
\begin{equation}
	\H_\parVec\coloneqq  \H - \sum_{i = 0} ^{m} \theta_i \mat{n}_i,
	\label{eq:sens_H_app}
\end{equation}
where $\parVec \coloneqq \{\theta_i\}_i$ is a set of real parameters to be estimated.

When the system undergoes linear dynamics, its evolution is fully captured by the response (Green's) function introduced above in \eqnref{eq:SAout}, and stated explicitly in \eqnref{eq:Greens}, which generally reads:
\begin{align}
	\smat{G}_\parVec[\omega] 
	& = \sform \left((\omega-\omega_0)\Id - \smat{H}_\parVec\right)^{-1} \label{eq:Greens_par_app}\\
	& = \sform \left((\omega-\omega_0)\Id - \smat{H} + \sum_{i = 0} ^{m} \theta_i \smat{n}_i\right)^{-1} 
	\label{eq:Greens_par_expl_app}
\end{align}
where $\omega$ is the frequency of the input/output modes being considered, while $\smat{H}_\parVec=\stransf{\mat{H}_\parVec}$ is the phase-space representation \eref{eq:phase_space_rep} of the dynamical generator \eref{eq:sens_H_app}, whose explicit form yields \eqnref{eq:Greens_par_expl_app} with each $\smat{n}_i=\stransf{\mat{n}_i}$. In general, the dependence on the internal frequency $\omega_0$ could be dropped above by redefining the generator as $\smat{H}'_\parVec \coloneqq \smat{H}_\parVec-\omega_0 \Id$.

Now, as Gaussian dynamics is considered, it is sufficient to determine the average amplitude and the covariance matrix of the modes being measured, $\svec{S}^{A}_\mrm{out}[\omega]$ and $\svec{V}^{A}_\mrm{out}[\omega]$ (we drop the superscript $A$ below for generality, as well as the $\omega$-dependence) that are specified by the input-output relations \eref{eq:Sout} and \eref{eq:Vout}, respectively, with the response function $\smat{G}_\parVec$ given by \eqnref{eq:Greens_par_app}. In particular, the output amplitude \eref{eq:Sout} and covariance matrix \eref{eq:Vout} are now $\parVec$-dependent and read 
\begin{align}
	\svec{S}(\parVec) &= \left(\Id - \kappa\smat{G}_{\parVec} \right) \svec{S}_\mrm{in}, \label{eq:SparVec} \\
	\smat{V}(\parVec)  &= \left(\Id - \kappa\smat{G}_{\parVec} \right) \smat{V}_\mrm{in}^A \left(\Id - \kappa\smat{G}_{\parVec} \right)^\T + \kappa \smat{G}_{\parVec} \Xi \tilde{\smat{V}}_\mrm{in}^{B} \Xi^\T \smat{G}_{\parVec}^\T,
	\label{eq:VparVec}
\end{align}
so the $\parVec$-parametrised version of the noisy QFIM  \eqref{eq:FmatrixNoisyApp} is 
\begin{align}
	\FQ_{jk} 
	&\approx 
	\frac{1}{2} \Tr\!\left[ \smat{V}(\parVec)^{-1} (\partial_{j} \smat{V}(\parVec)) \smat{V}(\parVec)^{-1} (\partial_k \smat{V}(\parVec))\right] \nonumber \\
	& + 
	(\partial_j \svec{S}(\parVec))^\T \smat{V}(\parVec)^{-1}  (\partial_k  \svec{S}(\parVec)).
	\label{eq:FmatrixNoisyparVecApp}
\end{align}

Computing derivatives of \eqnsref{eq:Greens_par_expl_app}{eq:SparVec} w.r.t.~a given parameter $\theta_i$, we obtain $\partial_i\smat{G}_{\parVec}=\smat{G}_{\parVec} \smat{n}_{i} \sform \smat{G}_{\parVec}$ and $\partial_i \svec{S}(\parVec) = - \kappa\,\partial_i\smat{G}_{\parVec}\svec{S}_\mrm{in}$, respectively, which imply that the corresponding derivative of the output amplitudes just reads
\begin{equation}
	\partial_i \svec{S}(\parVec) = - \kappa\,\smat{G}_{\parVec} \smat{n}_{i} \sform \smat{G}_{\parVec}\svec{S}_\mrm{in}.
	\label{eq:dS}
\end{equation} 

For the covariance matrix, however, as we are interested only in small perturbations around zero in \eqnref{eq:sens_H_app}, i.e.~$\theta_i\approx0$ at which the sensitivity is supposed to be unbounded, we assume that $\smat{G}_{\parVec}$ is then divergent---if otherwise, see below, the solution can be easily found case by case. As a result, only the terms quadratic in $\smat{G}_{\parVec}$ may be kept in \eqnref{eq:VparVec}, which simplifies then to $\smat{V}(\parVec) \approx \smat{G}_{\parVec} \smat{V}_\mrm{in}  \smat{G}_{\parVec}^\T$ and yields:
\begin{equation}
	\partial_i \smat{V}(\parVec) \approx \partial_i \smat{G}_{\parVec} \smat{V}_\mrm{in}  \smat{G}_{\parVec}^\T+\smat{G}_{\parVec} \smat{V}_\mrm{in}\partial_i \smat{G}_{\parVec}^\T
	\label{eq:dV}
\end{equation}
for an input covariance $\smat{V}_\mrm{in}$ (here, $\smat{V}_\mrm{in}= \kappa^2 \smat{V}_\mrm{in}^{A} + \kappa \Xi \tilde{\smat{V}}_\mrm{in}^{B} \Xi^\T$). 

Finally, substituting for the explicit forms of the output amplitudes and covariances (valid for $\theta_i\approx0$) as well as their derivatives \eref{eq:dS} and \eref{eq:dV}, respectively, we obtain the expression for the noisy QFIM \eref{eq:FmatrixNoisyparVecApp} that applies when estimating (small) linear perturbations around $\smat{G}_{\parVec}$-divergences as
\begin{align}
	\label{eq:FQjk}
	\FQ_{jk} 
	& \approx
	\Tr\!\left[ \smat{n}_{j} \sform   \smat{G}_{\parVec} \smat{n}_{k} \sform  \smat{G}_{\parVec} \right] +   \Tr[\smat{V}_\mrm{in}^{-1} \smat{n}_{j} \sform \smat{G}_{\parVec} \smat{V}_\mrm{in} \smat{G}_{\parVec}^\T \sform^\T \smat{n}_{k}^\T ] \nonumber\\
	& \quad+ 
	\svec{S}_\mrm{in}^\T  \smat{G}_{\parVec}^\T \sform^\T \smat{n}_{j}^\T  \smat{V}_\mrm{in}^{-1}   \smat{n}_{k} \sform \smat{G}_{\parVec}  \svec{S}_\mrm{in},
\end{align}
which also assumes the covariance matrix of the input modes $\smat{V}_\mrm{in}$ to be invertible---as is the case for any Gaussian state of finite energy~\cite{Ferraro2005}.

\section{Sain-Massey procedure for expanding inverses of singular matrices}
\label{app:SM}
Let $\A(z)= \A_0 + z \A_1 + z^2 \A_2 + \cdots$ be an analytic matrix-valued function of $z$ in some non-empty neighbourhood of $z = 0$, such that $A^{-1}(z)$ exists in some (possibly punctured) disc centered at $z=0$. Then, $\A^{-1}(z)$ possesses a Laurent series expansion 
\begin{equation}\label{eq:Ainv}
	\A^{-1}(z)   =   \frac{1}{z^s} \left( \X_0 + z \X_1 + z^2 \X_2 + \dots  \right),
\end{equation}
where $\X_0 \ne 0$ and $s$ is a natural number, known as the order of the pole at $z=0$. In order to determine $s$ and the coefficient-matrices $\X_i$, we follow the so-called \emph{Sain-Massey} (SM) procedure~\cite{Sain1969,Howlett1982,Howlett2001,avrachenkov2013analytic}. 

In particular, we define first the \emph{augmented matrix} as
\begin{equation}
	\bm{\mathcal{A}}^{(t)} := \begin{pmatrix}
		\A_0 & 0      & 0      & \dots & 0\\
		\A_1 &  \A_0 & 0     &  \dots & 0 \\
		\A_2 & \A_1  & \A_0 & \dots & 0 \\
		\vdots & \vdots & \vdots & \dots & \vdots \\
		\A_t & \A_{t-1} & \A_{t-2} & \dots & \A_0
	\end{pmatrix}.
	\label{eq:aug_matrix}
\end{equation}
The order of the pole $s$ is then given by the smallest value of $t$ for which $\rank\bm{\mathcal{A}}^{(t)}-\rank\bm{\mathcal{A}}^{(t-1)}=n$, where $n$ is the dimension of $\A(z)$. Furthermore, the (coefficient) matrices $\X_k$ are given by the following recursive formula for $k=1,2,\dots$:
\begin{equation}\label{eq:recur}
	\X_k  = \sum\limits_{j=0}^{s} \mathcal{G}_{0j}^{(s)} \left( \delta_{j+k,s} \Id   -  \sum\limits_{i=1}^{k}  \A_{i+j} \X_{k-i}\right),
\end{equation} 
with  $\X_0 = \mathcal{G}_{0s}^{(s)}$ being the top-right block of the matrix 
\begin{equation}
	\mathcal{G}^{(s)} := \begin{pmatrix}
		\mathcal{G}_{00}^{(s)}  & \dots & \mathcal{G}_{0s}^{(s)}\\
		\vdots            &     \dots & \vdots \\
		\mathcal{G}_{s0}^{(s) }  &  \dots &  \mathcal{G}_{ss}^{(s)}
	\end{pmatrix},
	\label{eq:Gmatrix}
\end{equation}
which is the Moore–Penrose generalized inverse of the augmented matrix at the pole level ($t=s$), i.e.~$\mathcal{G}^{(s)} := [\bm{\mathcal{A}}^{(s)}]^+$.

\section{Single-parameter perturbation sensing}
\label{app:single_par}
%
In case the aim is to sense a perturbation induced by only a \emph{single parameter}, $\theta_0$, the response function \eref{eq:Greens_par_expl_app} simplifies to
\begin{equation}
	\smat{G}_{\theta_0}[\omega=\omega_0] 
	= \sform \left(\theta_0 \smat{n}_0 - \smat{H}\right)^{-1} 
	\label{eq:response_function}
\end{equation}
where, as throughout the main text (see \eqnref{eq:Gomega}), we choose to be probing the system \emph{on resonance}, i.e.~we set $\omega=\omega_0$ and drop the $\omega$-dependence below for simplicity. However, see also Apps.~\ref{app:Det0}-\ref{app:lasing_threshold} below for important generalisations.

\subsubsection{Sensing at a non-singular point}
\label{app:single_par_nonsing}
Now, whenever the generator $\smat{H}$ is \emph{non-singular}, the response function \eref{eq:response_function} admits a Neumann series:
\begin{align}
	\smat{G}_{\theta_{0}} &= -\sform \smat{H}^{-1} \left[\Id + \sum_{k=1}^{\infty} \theta_{0}^k (\smat{n}_{0} \smat{H}^{-1})^{k}\right],	\nonumber \\
	&= -\smat{J} \smat{H}^{-1} \sum_{k=0}^{\infty} (\smat{n}_0 \smat{H})^k \theta_0^k = \sum_{k=0}^{\infty} \smat{g}_k \theta_0^k, 
	\label{eq:expnG}
\end{align}
where we have defined the expansion (matrix) coefficients as $\smat{g}_k \coloneqq -  \smat{J} \smat{H}^{-1} (\smat{n}_0  \smat{H}^{-1})^k$. As a result, the covariance matrix \eref{eq:VparVec} can also be expanded as follows:
\begin{align}
	\smat{V}(\parVec) &= (\Id - \kappa \smat{G}_{\theta_{0}}) \smat{V}^A_\mrm{in} (\Id - \kappa \smat{G}_{\theta_{0}})^\T + \kappa \smat{G}_{\theta_{0}} \Xi \tilde{\smat{V}}^B_\mrm{in} \Xi^\T \smat{G}_{\theta_{0}}^\T, \nonumber \\
	&=  \smat{V}^A_\mrm{in} - \kappa (\smat{G}_{\theta_{0}} \smat{V}_\mrm{in}^A + \smat{V}_\mrm{in}^A  \smat{G}_{\theta_{0}}^\T ) + \smat{G}_{\theta_{0}} \Lambda \smat{G}_{\theta_{0}}^\T, \nonumber \\
	&=              \smat{V}^A_\mrm{in} -  \kappa  \sum_{\ell =0}^{\infty} (\smat{g}_\ell  \smat{V}_\mrm{in}^A +  \smat{V}_\mrm{in}^A \smat{g}_\ell^\T ) \theta_0^\ell \nonumber \\&+ \sum_{m = 0}^{\infty}  (\smat{g}_0 \Lambda \smat{g}_m^\T + \smat{g}_1 \Lambda \smat{g}_{m-1}^\T + \cdots + \smat{g}_m \Lambda  \smat{g}_0^\T )\theta_0^m, \nonumber \\
	&=       \smat{V}^A_\mrm{in}  -     \sum_{n =0}^{\infty}    \smat{v}_n \theta_0^n, 
	\label{eq:expnV}                                                   
\end{align}
with the expansion coefficient now reading
\begin{equation}
	\smat{v}_n \coloneqq   - \kappa (\smat{g}_n  \smat{V}_\mrm{in}^A +  \smat{V}_\mrm{in}^A \smat{g}_n^\T ) + \smat{g}_0 \Lambda \smat{g}_n^\T + \smat{g}_1 \Lambda \smat{g}_{n-1}^\T + \cdots + \smat{g}_n \Lambda  \smat{g}_0^\T,
\end{equation}
where we have defined $\Lambda \coloneqq  \kappa^2  \smat{V}_\mrm{in}^A + \kappa  \Xi \tilde{\smat{V}}^B_\mrm{in} \Xi^\T$.

Now, evaluating the derivative w.r.t.~$\theta_0$ and/or taking the limit $\theta_0\to0$ of expressions \eqref{eq:SparVec}, \eqref{eq:expnG} and \eqref{eq:expnV}, we obtain: 
\begin{equation} 
	\begin{aligned}
		\lim_{\theta_0 \to 0}  \partial_{\theta_0} \smat{G}_{\theta_{0}} &= \smat{g}_1, \\
		\lim_{\theta_0 \to 0}  	\partial_{\theta_0} \smat{S}(\parVec) &= \lim_{\theta_0 \to 0}    -\kappa (\partial_{\theta_0} \smat{G}_{\theta_{0}} )\svec{S}_\mrm{in} = -\kappa \smat{g}_1  \svec{S}_\mrm{in}, \\
		\lim_{\theta_0 \to 0}   \smat{V}(\parVec) &= \smat{V}^A_\mrm{in}   -  \smat{v}_0, \\
		\lim_{\theta_0 \to 0}  \partial_{\theta_0}  \smat{V}(\parVec)& =   \lim_{\theta_0 \to 0} ( -\smat{v}_1 -2 \smat{v}_2 \theta_0 - \cdots ) =  - \smat{v}_1, \\
	\end{aligned}
\end{equation}
Hence, dealing here with a non-divergent response function \eref{eq:expnG}, we substitute these explicitly into \eqnref{eq:FmatrixNoisyparVecApp} (rather than \eqnref{eq:FQjk}), in order to obtain the valid noisy QFIM as
\begin{align}
	\FQ_{00} 
	&\approx
	\lim_{\theta_0 \to 0} \frac{1}{2} \Tr\!\left[ \smat{V}^{-1}(\parVec) \left(\partial_{\theta_0} \smat{V}(\parVec)\right) \smat{V}^{-1}(\parVec) \left( \partial_{\theta_0} \smat{V}(\parVec) \right) \right] \nonumber \\
	& +   
	\lim_{\theta_0 \to 0} 2 (\partial_{\theta_0}  \svec{S}(\parVec))^\T \smat{V}^{-1}(\parVec)  (\partial_{\theta_0} \svec{S}(\parVec)) \nonumber  \\
	&=
	C + \order{\theta_0},
	\label{eq:FQ00_nonsing}
\end{align}
with a constant
\begin{align}
	C 
	& \coloneqq 
	\frac{1}{2} \Tr\!\left[(\smat{V}^A_\mrm{in}   -  \smat{v}_0)^{-1} \smat{v}_1 (\smat{V}^A_\mrm{in}   -  \smat{v}_0)^{-1} \smat{v}_1 \right] \nonumber \\&+2 \kappa^2 \svec{S}_\mrm{in}^\T \smat{g}_1^\T (\smat{V}^A_\mrm{in}   -  \smat{v}_0)^{-1}  \smat{g}_1 \svec{S}_\mrm{in}.
	\label{eq:C_const}
\end{align}

Consequently, the QFI must approach a constant value $C<\infty$ as $\theta_0\to0$ and, hence, cannot be divergent. In other words, this sets a non-zero lower bound on the estimation error $\qerrorS\theta_{0} = 1/\sqrt{C}>0$ that applies at $\theta_0=0$, and precludes infinite precision in estimating small perturbations of $\theta_0$.

\subsubsection{Sensing at a singular point}
\label{app:single_par_sing}
When the generator $\smat{H}$ is \emph{singular} instead, the response function \eref{eq:response_function} becomes divergent at $\theta_0=0$, so we can directly use the formula for the noisy QFIM \eref{eq:FQjk} with $j=k=0$, i.e.:
\begin{align}
	\label{eq:FQ00}
	\FQ_{00} 
	& \approx
	\Tr\!\left[ \smat{n}_{0} \sform   \smat{G}_{\theta_0} \smat{n}_{0} \sform  \smat{G}_{\theta_0} \right] +   \Tr[\smat{V}_\mrm{in}^{-1} \smat{n}_{0} \sform \smat{G}_{\theta_0} \smat{V}_\mrm{in} \smat{G}_{\theta_0}^\T \sform^\T \smat{n}_{0}^\T ] \nonumber\\
	& \quad+ 
	\svec{S}_\mrm{in}^\T  \smat{G}_{\theta_0}^\T \sform^\T \smat{n}_{0}^\T  \smat{V}_\mrm{in}^{-1}   \smat{n}_{0} \sform \smat{G}_{\theta_0}  \svec{S}_\mrm{in},
\end{align}
where, however, we follow the Sain-Massey (SM) procedure described above in \appref{app:SM} to determine the correct small-$\theta_0$ expansion of the response function \eref{eq:response_function}. Although this must then be carried out for each particular case, the corresponding Laurent expansion of $\smat{G}_{\theta_0}$ generally takes the form:
\begin{equation}
	\smat{G}_{\theta_0}  = \sform \theta_0^{-s} \sum_{k=0}^r \theta_0^k \,\X_k(\smat{H},\smat{n}_0),
	\label{eq:G_SMexpansion_gen}
\end{equation}
with each (matrix) coefficient, $\X_k(\smat{H},\smat{n}_0)$, being non-zero up to some (possibly infinite) $r$. In the above, we emphasise that each coefficient in the expansion generally depends on both the dynamical generator $\smat{H}$ and the perturbation matrix $\smat{n}_0$ in \eqnref{eq:response_function}, while the singularity is then characterized by the order $s\in \mathbb{N}_{+}$ of the pole determined during the SM procedure. 

Still, assuming the general form \eref{eq:G_SMexpansion_gen} of the singular expansion, one may evaluate the form of the noisy QFI \eref{eq:FQ00} in the limit $\theta_0\to0$, which then reads
\begin{equation}
	\FQ_{00} \approx \theta_0^{-2s} \left[ A + \order{\theta_0 }  \right]
	\label{eq:FQ00_sing}
\end{equation}
with a constant
\begin{align}
	A 
	& \coloneqq
	\Tr[\smat{n}_0 \X_0 \smat{n}_0 \X_0 + \smat{V}_\mrm{in}^{-1} \smat{n}_0 \X_0 \smat{V}_\mrm{in} \X_0^\T \smat{n}_{0}^\T]\nonumber \\
	& \quad+ 
	\kappa^2 \svec{S}_\mrm{in}^\T  \X_0^\T    \smat{n}_0^\T  \smat{V}_\mrm{in}^{-1} \smat{n}_0   \X_0  \svec{S}_\mrm{in},
\end{align}
being determined by the zeroth-order coefficient, $\X_0(\smat{H},\smat{n}_0)$, in the expansion \eref{eq:G_SMexpansion_gen}. Crucially, thanks to the singularity of $\smat{H}$, the QFI diverges now at a rate $1/\theta_0^{2s}$ as $\theta_0\to0$, what proves that the sensitivity to $\theta_0$-perturbations is then infinite. In particular, the estimation error $\qerrorS\theta_{0}  = 1/\sqrt{\FQ_{00}} \propto \theta_{0}^{s}$ vanishes now as $\theta_0\to0$ at a rate dictated by the pole order $s$.

\bigskip

\subsection{Singularity-induced divergences for particular perturbations}
\label{app:SM_examples}
In what follows, we follow the SM procedure of \appref{app:SM} to determine explicitly the Laurent expansions of the response function $\smat{G}_{\theta_{0}}$, which we term \emph{SM expansions} for short, for the two-cavity system of interest and particular choices of the perturbations $\smat{n}_0$ in \eqnref{eq:response_function} (or equivalently $\mat{n}_0$ in the two-mode form of Langevin equations), which are listed in \tabref{tab:scaling}.

Firstly, by inspecting \eqnref{eq:response_function}, we observe that the expansion of $\smat{G}_{\theta_{0}}$ in $\theta_0$ resembles the one of \eqnref{eq:Ainv} with $\A(\theta_0) = \theta_0 \smat{n}_0-\smat{H}$, so that $\A_0 = -\smat{H}$, $\A_1 = \smat{n}_0$ and $\forall_{\ell>1}:\,\A_\ell = \smat{0}$. In particular, as $\det\smat{H}=0$, the response function must admit an SM expansion of the form:
\begin{equation}
\label{eq:SME}
\smat{G}_{\theta_0}  = \sform \theta_0^{-s} \left( \X_0 + \theta_0 \X_1 + \theta_0^2 \X_2 + \cdots \right),
\end{equation} 
whose coefficients are decided by the forms of $\smat{H}$ and $\smat{n}_{0}$. 

In what follows, we consider $\smat{H}$ not only to be singular but also fulfil the PT symmetry, so that it corresponds to $\H$ of \eqnref{eq:sys_H} with $g=\gamma_1=\gamma_2=1$~\cite{Note9}, i.e.~$\bar{\H}$ introduced in \eqnref{eq:HS_HNS}, whose form in phase-space reads 
\begin{equation}\label{eq:Hsing}
	\bar{\smat{H}}  = \begin{pmatrix}
		0       & 1 & 1  & 0      \\
		1  &  0     &   0     & -1 \\
		-1 &  0     &   0     & 1   \\
		0       & 1 &  1 & 0
	\end{pmatrix}.
\end{equation}

As the first example, we consider the two-mode symmetric perturbation characterised by 
\begin{equation}
	\smat{n}_0 = \Id =\begin{pmatrix}
		1	 &   0   &   0    &  0  \\
		0   &   1  &    0    &  0  \\
		0   &   0   &   1    &  0  \\
		0   &   0   &   0    &  1  
	\end{pmatrix},
	\label{appn:n2mode}
\end{equation} 
which corresponds to varying the frequencies of both cavities in an identical manner~\cite{LJ}. Then, it is sufficient to compute only the first three augmented matrices \eref{eq:aug_matrix}:
\begin{align}
	\bm{\mathcal{A}}^{(0)} &= \bpm -\bar{\smat{H}} \epm, \\
	\bm{\mathcal{A}}^{(1)} &= \bpm -\bar{\smat{H}} & 0 \\ \Id & -\bar{\smat{H}} \epm, \\
	\bm{\mathcal{A}}^{(2)} &=  \bpm -\bar{\smat{H}} & 0 & 0 \\ \Id & -\bar{\smat{H}} & 0 \\ 0 & \Id & -\bar{\smat{H}} \epm,
	\label{eq:augA2}
\end{align}
as one may then verify that $\rank [\bm{\mathcal{A}}^{(2)}] - \rank[\bm{\mathcal{A}}^{(1)}] = 4$, i.e.~the dimension of $\smat{G}_{\theta_0}$, so that the pole in the SM expansion \eref{eq:SME} is of order two, $s = 2$. Furthermore, the coefficients of the SM expansion \eref{eq:SME} can now be determined by computing the matrix $\mathcal{G}^{(s)}$ defined in \eqnref{eq:Gmatrix} with $s=2$ that is given by the Moore-Penrose inverse of $\bm{\mathcal{A}}^{(2)}$. In particular, we substitute explicitly for $\bar{\smat{H}}$ specified in \eqnref{eq:Hsing} into \eqnref{eq:augA2}, in order to obtain the relevant part of the matrix $\mathcal{G}^{(s)}$ as
\begin{widetext}
	\begin{equation}
		\mathcal{G}^{(2)} = [\bm{\mathcal{A}}^{(2)}]^+ =  
		\left(
		\begin{array}{c|c|c}
			\frac{1}{5}\bpm 0& -1& 1 & 0 \\ -1 & 0 & 0& -1 \\ -1 & 0 & 0& -1 \\ 0& 1 & -1 &0 \epm & \frac{1}{5}\bpm 3 & 0 & 0 & 2\\ 0 & 3 & -2 & 0 \\ 0& -2 & 3 & 0\\ 2 & 0 & 0 & 3\epm   &\bpm  0 & 1 & 1 & 0 \\ 1 & 0 & 0 & -1 \\ -1 & 0 & 0 & 1 \\ 0 & 1 & 1 & 0   \epm \\
			\hline
			\cdots&\cdots & \cdots
		\end{array}
		\right)_{12 \times 12},
		\label{eq:MPinv2modes}
	\end{equation}
\end{widetext} 
whose top-right block, see \eqnref{eq:Gmatrix}, yields the coefficient $\smat{X}_0$ in \eqnref{eq:SME} being actually equal to $\bar{\smat{H}}$, i.e.:
\begin{equation}\label{eq:X0_2mode}
	\smat{X}_0 = \mathcal{G}_{02}^{(2)} = \begin{pmatrix}
		0       & 1 & 1  & 0      \\
		1  &  0     &   0     & -1 \\
		-1 &  0     &   0     & 1   \\
		0       & 1 &  1 & 0
	\end{pmatrix} = \bar{\smat{H}}.
\end{equation}
Now, using the recurrence relation \eref{eq:recur} we obtain
\begin{align}
	\X_1 &= - \mathcal{G}_{00}^{(2)} \A_1 \X_0  + \mathcal{G}_{01}^{(2)}, \nonumber \\
	\X_2 &= \mathcal{G}_{00}^{(2)} (\Id - \A_1 \X_1), \nonumber \\
	\forall_{\ell \ge 3}:\quad\X_\ell &=  -  \mathcal{G}_{00}^{(2)}  \A_1 \X_{\ell-1},
\end{align}
so that substituting $ \mathcal{G}_{00}^{(2)} $, $ \mathcal{G}_{01}^{(2)} $ from \eqnref{eq:MPinv2modes} and noting that $\A_1 = \Id$, we find $\X_1 = \Id$ and $\X_\ell = 0$ for $\ell \ge 2$. Hence, the full SM expansion of the response function in presence of two-mode symmetric perturbation \eref{appn:n2mode} reads
\begin{equation}
	\smat{n}_0 = \Id: \qquad
	\smat{G}_{\theta_0}  = \sform \theta_0^{-2} \left( \bar{\smat{H}} + \theta_0 \Id \right).
	\label{eq:SME_2mode}
\end{equation}

As a second example, let us consider a one-mode perturbation described by
\begin{equation}\label{appn:n1mode}
	\smat{n}_0 = \begin{pmatrix}
		1	 &   0   &   0    &  0  \\
		0   &   0  &    0    &  0  \\
		0   &   0   &   1    &  0  \\
		0   &   0   &   0    &  0  
	\end{pmatrix} =: \tilde{\Id},
\end{equation} 	
which corresponds to modifying the frequency of only the first cavity~\cite{LJ}. Then, it is sufficient to compute only the first two augmented matrices:
\begin{align}
	\bm{\mathcal{A}}^{(0)} &= \bpm -\bar{\smat{H}} \epm, \\
	\bm{\mathcal{A}}^{(1)} &= \bpm -\bar{\smat{H}} & 0 \\ \tilde{\Id} & -\bar{\smat{H}} \epm,
\end{align}
with $\rank [\bm{\mathcal{A}}^{(1)}] - \rank[\bm{\mathcal{A}}^{(0)}] = 4$, leading to a pole of order one ($s=1$) in the SM expansion. Now, the (matrix) coefficient $\X_0=\mathcal{G}_{01}^{(1)}$ corresponds to the top-right block of the Moore-Prenrose inverse of the augmented matrix $\bm{\mathcal{A}}^{(1)}$, i.e.:
\begin{equation}
	\mathcal{G}^{(1)}  =  
	\left(
	\begin{array}{c|c}
		\frac{1}{2}\bpm 0&0&0&0 \\ -1 & 0 & 0& -1 \\ 0 &0& 0&0 \\ 0& 1 & -1 &0 \epm & \bpm 1 &0&0&-1\\ 0&-1&-1&0\\0&1&1&0\\1&0&0&-1\epm   \\
		\hline
		\cdots& \cdots
	\end{array}
	\right)_{8 \times 8}\!\!,
	\label{eq:MPinv1mode}
\end{equation}
while the remaining coefficients can be calculated by using the recurrence relation \eqref{eq:recur}, as in the previous case. We obtain
\begin{equation}
	\X_0 = \begin{pmatrix}
		1 &  0  & 0  & -1\\
		0 & -1  & -1 & 0 \\
		0 &  1   & 1  & 0 \\
		1 &  0   & 0  & -1
	\end{pmatrix}, \quad 
	\X_1 = \begin{pmatrix}
		0  & 0   & 0  &  0\\
		0  & 0  & 0 & -1\\
		0  & 0   & 0  & 0 \\
		0  & 1   & 0  & 1
	\end{pmatrix},                  
	\label{eq:X0X1}    
\end{equation}
and $\X_\ell = \smat{0}$ for $\ell \ge 2$, so that the full SM expansion of the response function in presence of one-mode symmetric perturbation \eref{appn:n1mode} reads
\begin{equation}
	\smat{n}_0 = \tilde{\Id}:\qquad
	\smat{G}_{\theta_0}  = \sform \theta_0^{-1} \left( \X_0 + \theta_0 \X_1 \right)
\end{equation} 
with the matrix coefficients specified in \eqnref{eq:X0X1}.

One can proceed similarly for other examples of perturbation matrix $\smat{n}_0$, in particular, the ones corresponding to the scenarios listed below \eqnref{eq:sens_H}:~two-mode asymmetric variation of the cavity frequencies~\cite{hongkong} or (reciprocal and non-reciprocal) perturbations  of the coupling strength~\cite{AC1}, which both yield SM expansions with a pole at the order $s=1$. We summarise the values of pole orders ($s$) arising in all the mentioned cases in \tabref{tab:scaling}.

\begin{table}[t!]
	\centering
	{\begin{tabular}{|M{0.5cm}|p{5.5 cm}|M{2cm}|} 
			\hline \vspace{1mm}
			No.  &  Perturbation matrix $\mat{n}_{0}$ in Eq. (5) ($\bar{\parVec}=0$) &  Pole order: $s$ \\
			\hline  
			1 &	 $\begin{pmatrix}
				1 & 0 \\0  & 1
			\end{pmatrix}$:   two-mode symmetric                         &   2       \\
			\hline         
			2  &	 $\begin{pmatrix}
				1 & 0 \\0  & -1
			\end{pmatrix}$:    two-mode asymmetric              &   1     \\
			\hline       
			3 &	  $\begin{pmatrix}
				1 & 0 \\ 0  & 0
			\end{pmatrix}$: one-mode             &     1    \\
			\hline    
			4  &	  $\begin{pmatrix}
				0 & 1 \\ 1 & 0
			\end{pmatrix}$:  coupling strength $g$           &  1    \\
			\hline    
			5  &	  $\begin{pmatrix}
				0 & 1 \\ 0 & 0
			\end{pmatrix}$:  non-reciprocal component of $g$    &  1 \\
			\hline    
		\end{tabular}
	}
	\caption{\textbf{Pole orders} in the SM expansion \eref{eq:SME} of the response function \eref{eq:response_function}, arising when sensing linear perturbations of a singular PT-symmetric dynamical generator, i.e.~$\bar{\smat{H}}$ specified in \eqnref{eq:Hsing}, being generated by each of the listed matrices $\smat{n}_0=\stransf{\mat{n}_0}$ in \eqnref{eq:response_function}.}
	\label{tab:scaling}
\end{table}

\subsubsection{Comparison with previous results of \citeref{LJ}}
\label{app:comp_to_LJ}
For completeness, let us comment how the above general methods can only under special circumstances agree with the approach used in \citeref{LJ}, which is based on the Neumann series $(\Id - \smat{T})^{-1} = \sum_{k=0}^\infty \smat{T}^k$. When the perturbation matrix $\smat{n}_0$ is invertible, one is then tempted to expand the response function \eref{eq:response_function} w.r.t.~$\theta_0^{-1}$ instead, so that
\begin{align} \label{eq:G}
	\smat{G}_{\theta_0} &= \sform \left(\theta_0 \smat{n}  - \smat{H} \right)^{-1} = \sform \theta^{-1} \smat{n}^{-1} \left[ \Id - \underbrace{\theta_0^{-1} (\smat{H} \smat{n}^{-1}) }_{\smat{T}} \right]^{-1}  \nonumber \\&=   \sform \theta_0^{-1} \smat{n}^{-1} \sum\limits_{k=0}^{\infty}  \theta_0^{-k} (\smat{H}\smat{n}^{-1})^k.
\end{align}
Hence, in case of $\smat{n}_0=\Id$ and $\smat{H}=\bar{\smat{H}}$, as in \eqnref{eq:Hsing}, for which $\forall_{\ell\ge2}:\,\bar{\smat{H}}^\ell = \smat{0}$, we obtain
\begin{align}\label{eq:NS}
	\smat{G}_{\theta_0}(\smat{n}_0 = \Id)  = \sform \theta_0^{-2} \left( \bar{\smat{H}} + \theta_0 \Id \right),
\end{align}
which coincides with the correct SM expansion \eref{eq:SME_2mode}. The above construction cannot be applied in case of a non-invertible $\smat{n}_0$---in particular, when considering single-mode perturbations with $\smat{n}_0=\tilde{\Id}$ of \eqnref{appn:n1mode}---but it may still lead to incorrect conclusions for invertible perturbation matrices $\smat{n}_0$. 

Let us consider the case when it is the coupling strength $g$ in \eqnref{eq:H0_smatrix} that is perturbed, so that
\begin{equation}
	\smat{n}_0 = \begin{pmatrix}
		0	 &   1   &   0    &  0  \\
		1   &   0  &    0    &  0  \\
		0   &   0   &   0    &  1  \\
		0   &   0   &   1    &  0  
	\end{pmatrix}\!.
\end{equation} 	
Evaluating explicitly the terms in the Neumann series \eref{eq:G}, we have that for all $k=1,2,3 \dots$:
\begin{equation}
	\smat{n}_0^{-1} \left( \bar{\smat{H}} \smat{n}_0^{-1} \right)^k = 2^{k-1} \begin{pmatrix}
		0 & 1 & -1 & 0 \\
		1 & 0 &  0  & 1\\
		1 & 0 &  0  & 1\\
		0 & -1 & 1  & 0
	\end{pmatrix} \ne \smat{0},
\end{equation}
and so the series \eref{eq:G} never terminates. This is a consequence of the Neumann series being then \emph{divergent}:~$(\Id - \smat{T})^{-1} = \sum_{k=0}^\infty \smat{T}^k$ converges only if $\lambda_\mrm{max}(\smat{T}) < 1$, i.e.~the spectral radius of $\smat{T}$ is less than 1, while from \eqnref{eq:G} it follows that the spectrum of $\smat{T}$ diverges as $\theta_0 \rightarrow 0$. Hence, this leaves the rate of divergence ill-defined, while the SM expansion \eref{eq:SME} correctly predicts $\smat{G}_{\theta_0} \propto \theta_0^{-1}$ with pole of the order $s=1$, see \tabref{tab:scaling}.

\begin{figure*}[t]
	\newcommand\ww{59mm}
	\centering
	\includegraphics[width=\ww]{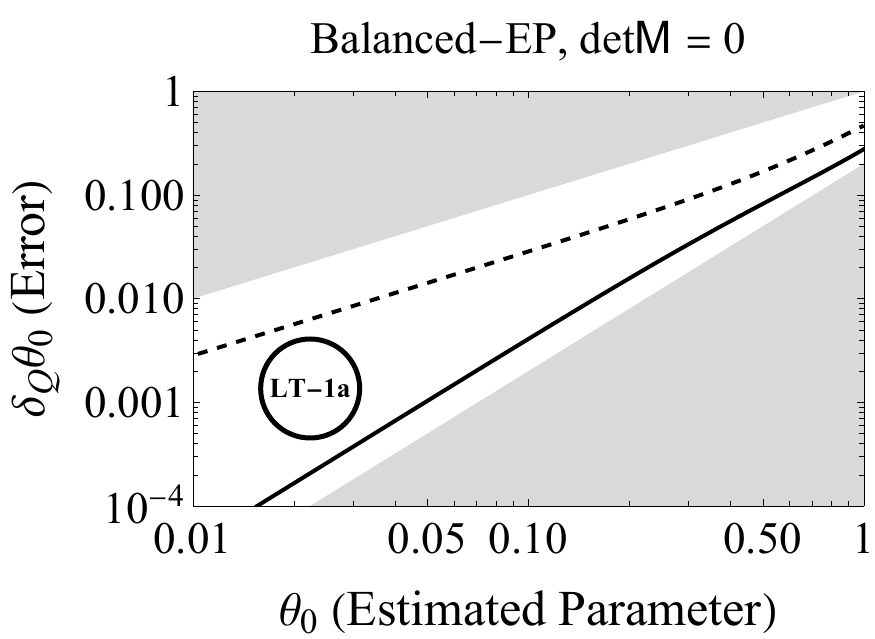}
	\includegraphics[width=\ww]{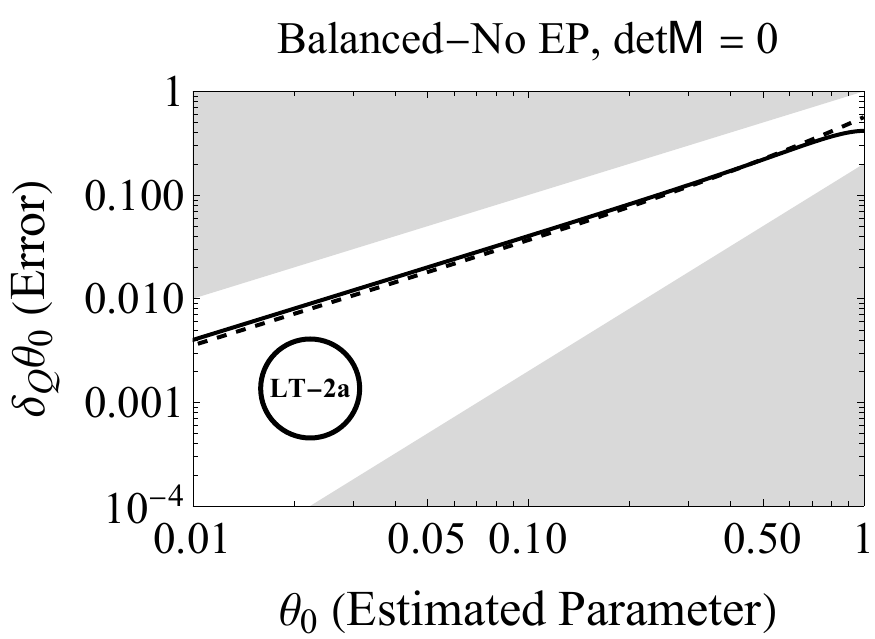}
	\includegraphics[width=\ww]{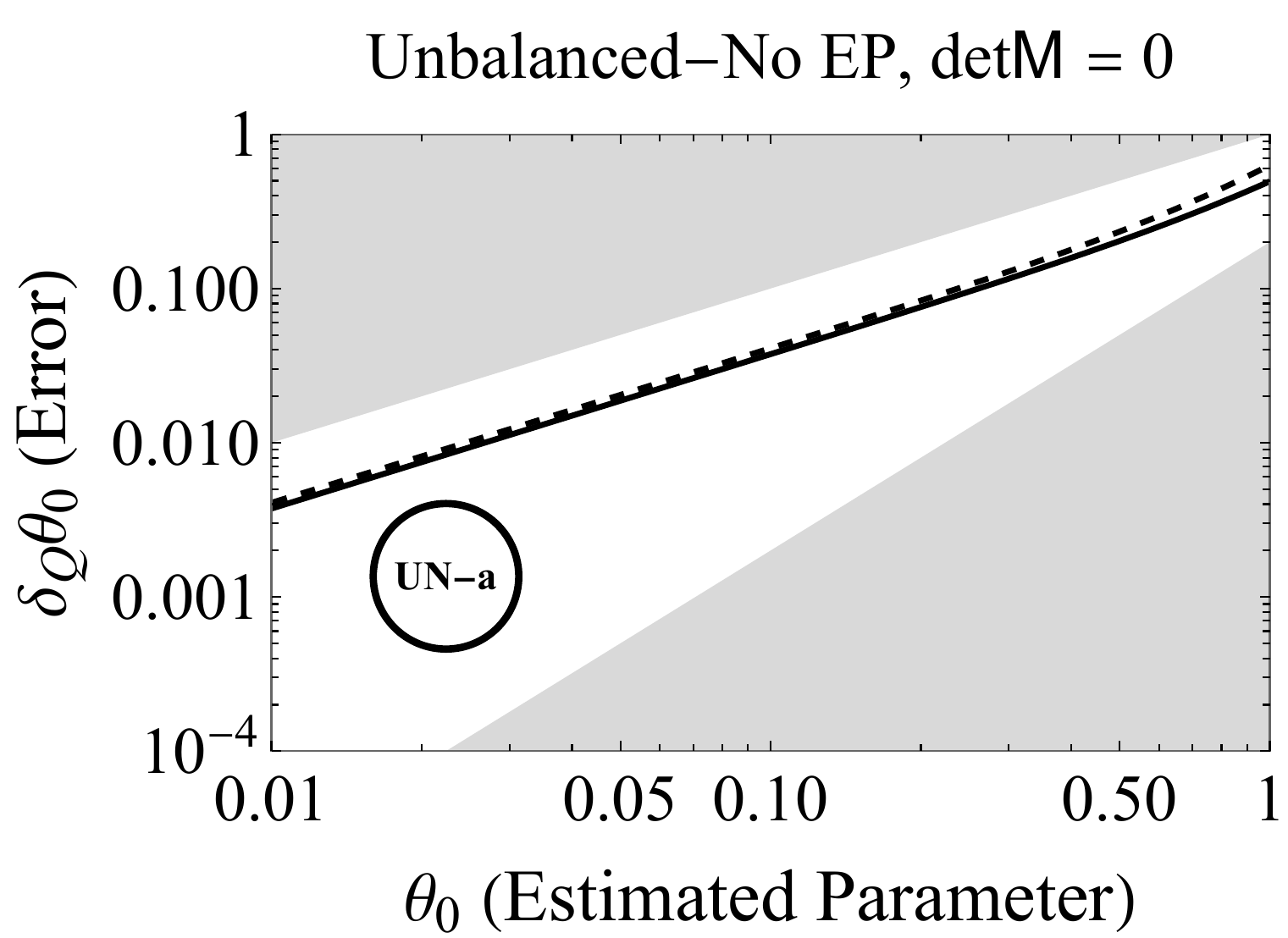}\\
	\includegraphics[width=\ww]{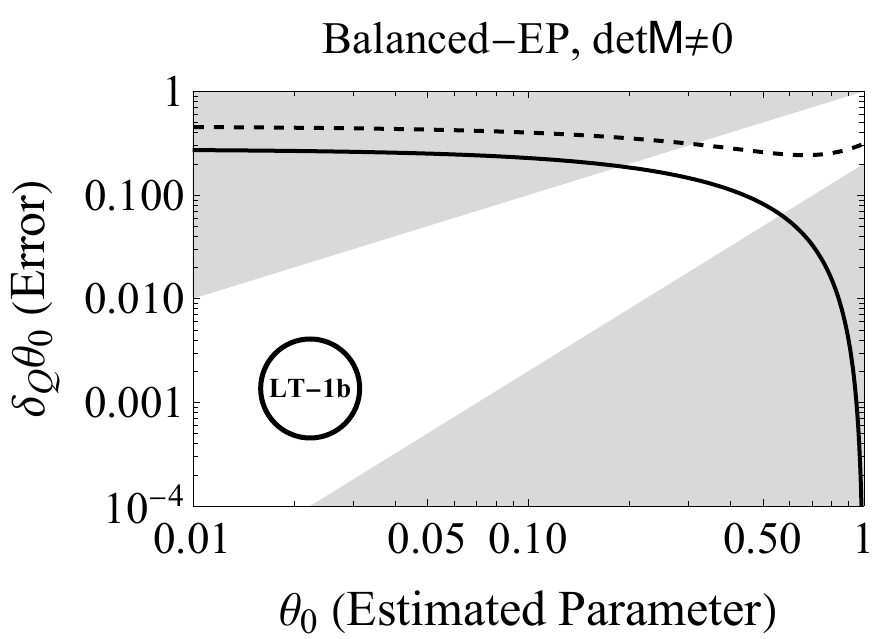}
	\includegraphics[width=\ww]{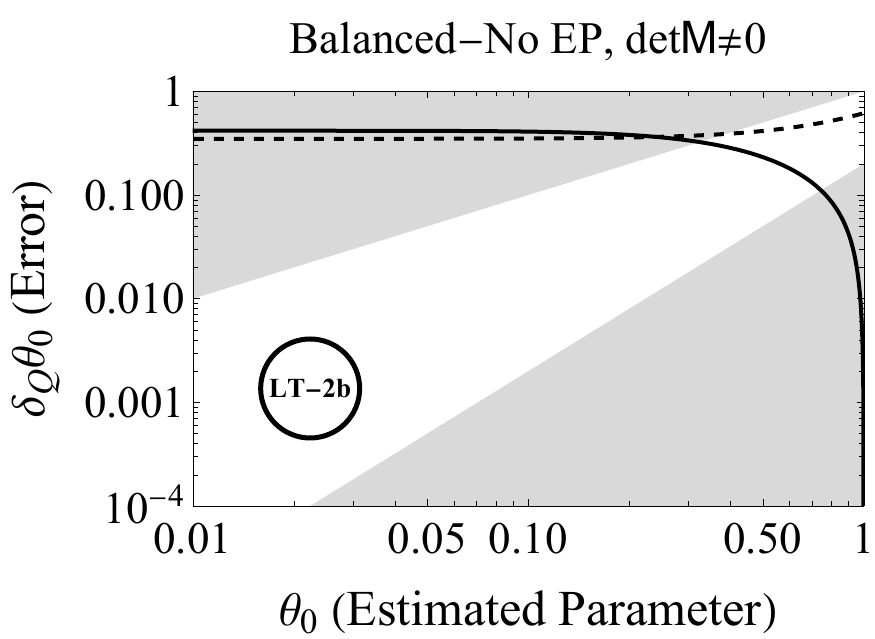}	
	\includegraphics[width=\ww]{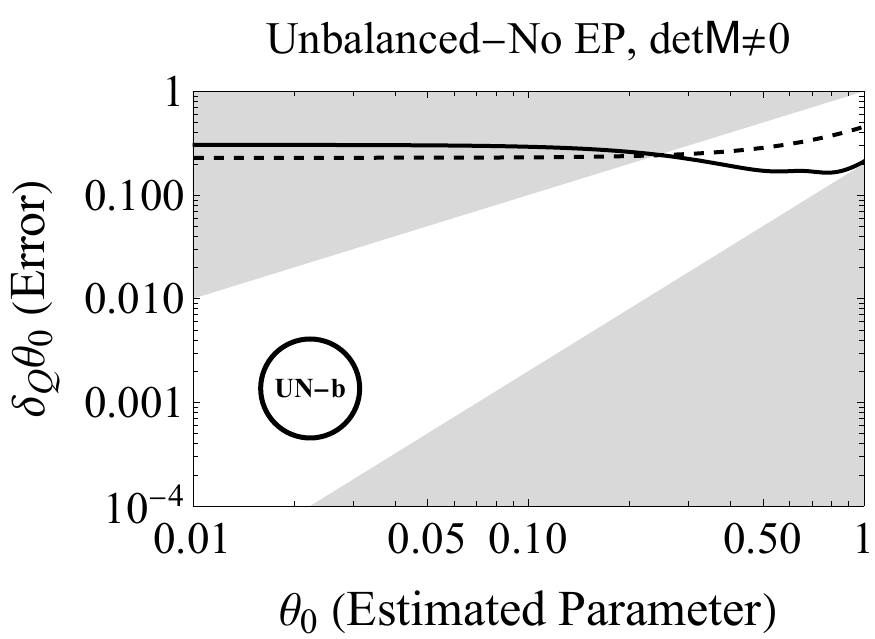}
	\caption{\textbf{Error in estimating the single-parameter $\theta_0$} (with others perfectly known) that describes two-mode or one-mode perturbations of cavity frequencies---solid and dashed lines, respectively---for different choices of the sensor working conditions. The specific choices of dynamical parameters, for which the effective generator $\smat{M} = \Id(\omega-\omega_0)-\smat{H}$ is \emph{singular} or \emph{non-singular}, are listed in \tabref{tab:EP_Sing_Bal}. The unbounded precision, i.e., the vanishing of error as $\theta_{0} \to 0$, is observed only when $\det \smat{M} = 0$, irrespective of whether the system is operated at an \emph{exceptional point} (EP) and/or \emph{balanced} condition is satisfied. However, if the EP and balanced conditions are fulfilled on top of $\det \smat{M} = 0$, the error shows a quadratic behaviour in case of the two-mode perturbation. Note that in the lower-left and lower-middle plots, the error vanishes as $\theta_{0} \to 1$ for the two-mode perturbation, which in this case actually corresponds to a singular point with $\det\smat{M}=0$. For completeness, see \appref{app:lasing_threshold} for a detailed discussion, we also mark above by the labels LT and UN the cases when the system is being perturbed at the \emph{lasing threshold} or within the \emph{unstable} phase (i.e.~does not possess a steady state in the long-time limit), respectively.
	}
	\label{fig:Det0}
	\vspace{1cm}
\end{figure*}
\begin{table*}
	\centering
	{\begin{tabular}{|M{2.0cm}||p{3.5cm}|p{3.5cm}|p{3.5cm}|} 
			\hline
			\textbf{Configuration}                     &    Balanced \& EP &   Balanced \& No-EP & Unbalanced \& No-EP \\
			\hline \vspace{3mm} 
			$\det \smat{M}=0 $                         &   $\omega=\omega_0$, \hfill \circled{LT-1a}  \newline $g =\gamma$, $\gamma_1=\gamma_2=\gamma$	&          $\omega=0$, $\omega_0 = \gamma$, \hfill  \circled{LT-2a}  \newline $g=\sqrt{2} \gamma$,  $\gamma_1=\gamma_2 =\gamma$  &   $\omega=\omega_0$ \hfill \circled{UN-a} \newline $g=2\gamma$, $\gamma_1 = \gamma$, $\gamma_2=4\gamma$ \\
			\hline    \vspace{3mm}       
			$\det \smat{M} \ne 0$                   &            $\omega=0$, $\omega_0 = \gamma$, \hfill \circled{LT-1b}  \newline  $g =\gamma$, $\gamma_1=\gamma_2 = \gamma $  &      $\omega = \omega_0$ \hfill  \circled{LT-2b} \newline $g=\sqrt{2} \gamma$,  $\gamma_1=\gamma_2 =\gamma$ 	 &  $\omega = \omega_0$ \hfill \circled{UN-b} \newline  $g=\gamma$, $\gamma_1 = \gamma$, $\gamma_2=\frac{\gamma}{2}$ \\
			\hline    
	\end{tabular}	}
	\caption{\textbf{Choices of parameters} determining the system dynamical generator $\smat{H}$, defined in \eqnref{eq:H0_smatrix}, that are used in various configurations in \figref{fig:Det0}, distinguishing whether the effective dynamical generator $\smat{M} =  \Id(\omega-\omega_0)-\smat{H}$ is singular or non-singular. }
	\label{tab:EP_Sing_Bal}
\end{table*}

\subsection{Singularity vs EP and balanced conditions}
\label{app:Det0}
The unbounded scaling of sensitivity in the described system arises when perturbing away from the singular dynamics of the sensor. In the main part of the paper, we have considered probing the signal output at the natural frequency of the sensor, i.e.~$\omega=\omega_0$. In this case, the condition for singular dynamics, i.e.~divergence of the response function \eref{eq:Greens}, is that the generator of sensor dynamics $\mat{H}$, or equivalently $\smat{H}$, constitutes a singular matrix. This then leads to divergent sensitivity (i.e.~the Fisher information) as $\theta_{0} \rightarrow 0$. Here, we illustrate that choosing singular operating conditions for the sensor to yield enhanced signal scaling with the sensed parameter is, in general, independent of tuning to an \emph{exceptional point} (EP) and/or \emph{balanced} conditions, as proposed earlier in the literature~\cite{LJ}.

In order to illustrate this, we focus on the two examples of sensor perturbations, i.e.~one-mode and two-mode (symmetric) perturbations of the cavity frequencies, as described above in \appref{app:SM} and \tabref{tab:scaling}, corresponding to $\mat{n}_0=(1,0;0,0)$ or $\mat{n}_0=\id$, respectively. Going beyond the results presented in the main part, we generalize the analysis beyond the case of probing at the sensor frequency, i.e.~we allow for $\omega\neq\omega_0$, in which case the response function can still be made divergent despite $\det \mat{H}\neq0$, as long as the system is PT-symmetric, i.e.~balanced $\gamma_1=\gamma_2=\gamma$ with the (strong) coupling being then fixed to $g=\sqrt{\gamma^2+(\omega-\omega_0)^2}>\gamma$. In particular, we consider the full dynamical generator appearing in \eqnref{eq:Greens_par_expl_app}:
\begin{equation}
	\smat{M} =   (\omega-\omega_0)\Id-\smat{H},
	\label{eq:Mgenerator}
\end{equation}
whose inverse now determines the response function $\smat{G}[\omega]$, see~\eqnref{eq:Greens} (or  \eqnref{eq:Greens_par_app} for $\parVec=\bm{0}$).

As seen in \figref{fig:Det0}, the vanishing behaviour of the error as $\theta_0\to0$ is completely lost if one operates in the non-singular regime, even when the EP and balance conditions are satisfied (see lower plots with $\det \smat{M}\neq0$). On the other hand, neither EP nor balance condition is actually needed to observe the unbounded linear scaling as long as the sensor parameters are tuned to the singularity (see upper plots with $\det \smat{M}=0$ and, in particular, the top-right plot UN-a). However, in case of the two-mode perturbation $\mat{n}_0=\id$, the simultaneous fulfilment of both EP and balanced condition enhances the sensitivity to scale quadratically with $\theta_0$ (see the top-left plot LT-1a) instead of being linear, as in all other cases. The choice of parameters for each of the subfigures in \figref{fig:Det0} is stated in \tabref{tab:EP_Sing_Bal}.

Lastly, we note that for all the cases with $\det \smat{M}=0$ listed in \tabref{tab:EP_Sing_Bal}, if one was to perform Jordan decomposition of the corresponding $\smat{M}$ matrices, the size of the largest block with zero eigenvalues is 2. Hence, it cannot be directly associated with $s$ in the $\theta_0^s$ scaling of the error---i.e.~the pole order of the SM expansion \eref{eq:SME}---as proposed in \citeref{LJ}, with contradiction being provided by the latter two cases, in which the error scales linearly ($s=1$) for either $\mat{n}_0$-perturbation.

\subsection{The role of the lasing threshold}
\label{app:lasing_threshold}
As discussed in the main text, the system dynamics \eref{eq:a1a2} (or equivalently \eqnref{eq:dyn_model}) is unstable when the imaginary parts of the eigenvalues of the dynamical generator $\H$ (see \eqnref{eq:sys_H})    are positive. Given a particular value of loss in the first cavity, $\gamma_1$, and relatively small inter-cavity coupling satisfying then $g\le\gamma_1$, this occurs when the gain in the second cavity is raised above $\gamma_2\ge g^2/\gamma_1$. In contrast, when the coupling between the cavities is relatively strong, i.e.~$g\ge\gamma_1$, it is sufficient to drive the second cavity at the same rate as the loss, i.e.~for any $\gamma_2\ge\gamma_1$ the system becomes then unstable. In particular, in the former case the transition point occurs at the value of $\gamma_2$ at which the system is, in fact, \emph{singular} (i.e.~$\det\H=0$ or $g^2=\gamma_1\gamma_2$), while in the latter when the system is \emph{balanced} ($\gamma_2=\gamma_1\eqqcolon\gamma$) or more precisely \emph{PT-symmetric}, as $g\ge\gamma$ is then automatically fulfilled. We illustrate these two ``phases" explicitly in \figref{fig:lasing}, which should be interpreted as the 2D parameter-space $\{\gamma_2,g\}$ plotted for a given value of loss $\gamma_1$ (equivalent to \figref{fig:par_space}\textbf{b}).  
	
From the physical perspective, the instability of the dynamics corresponds to both cavities exhibiting \emph{lasing}~\cite{Peng2014b}, while the aforementioned transition point corresponds to the so-called \emph{lasing threshold} (LT), above which (in the sense of increasing the gain $\gamma_2$) the excitations in both cavities grow limitlessly over time and with the system being \emph{unstable} (UN)---unless one includes other (e.g.~higher-order non-linearities) dissipative processes that are neglected by the linear model. 

In what follows, we provide a simplified proof that such a dynamical behaviour is indeed exhibited if one was to track the mean number of photons in the cavities, but more careful derivations based on the Langevin formalism yield same conclusions~\cite{NewProject}. For completeness, we further demonstrate that at the LT one may always make the sensitivity diverge by tailoring the effective dynamical generator in \eqnref{eq:Mgenerator} to be singular, $\det\smat{M}=0$, via tuning the probing frequency away from the resonance, i.e.~allowing for $\omega\neq\omega_0$. However, as the unbounded sensitivity may also be induced by the singularity away from the LT, within the UN phase, we note that operating at the LT is formally a \emph{sufficient but not necessary} condition. Finally, we explicitly show that, if probing the system at the LT, the unbounded sensitivity may be exhibited no matter whether the (singular) system is perturbed away from the LT either into the stable or the unstable dynamical phase.

\begin{figure}[t]
	\includegraphics[width=\columnwidth]{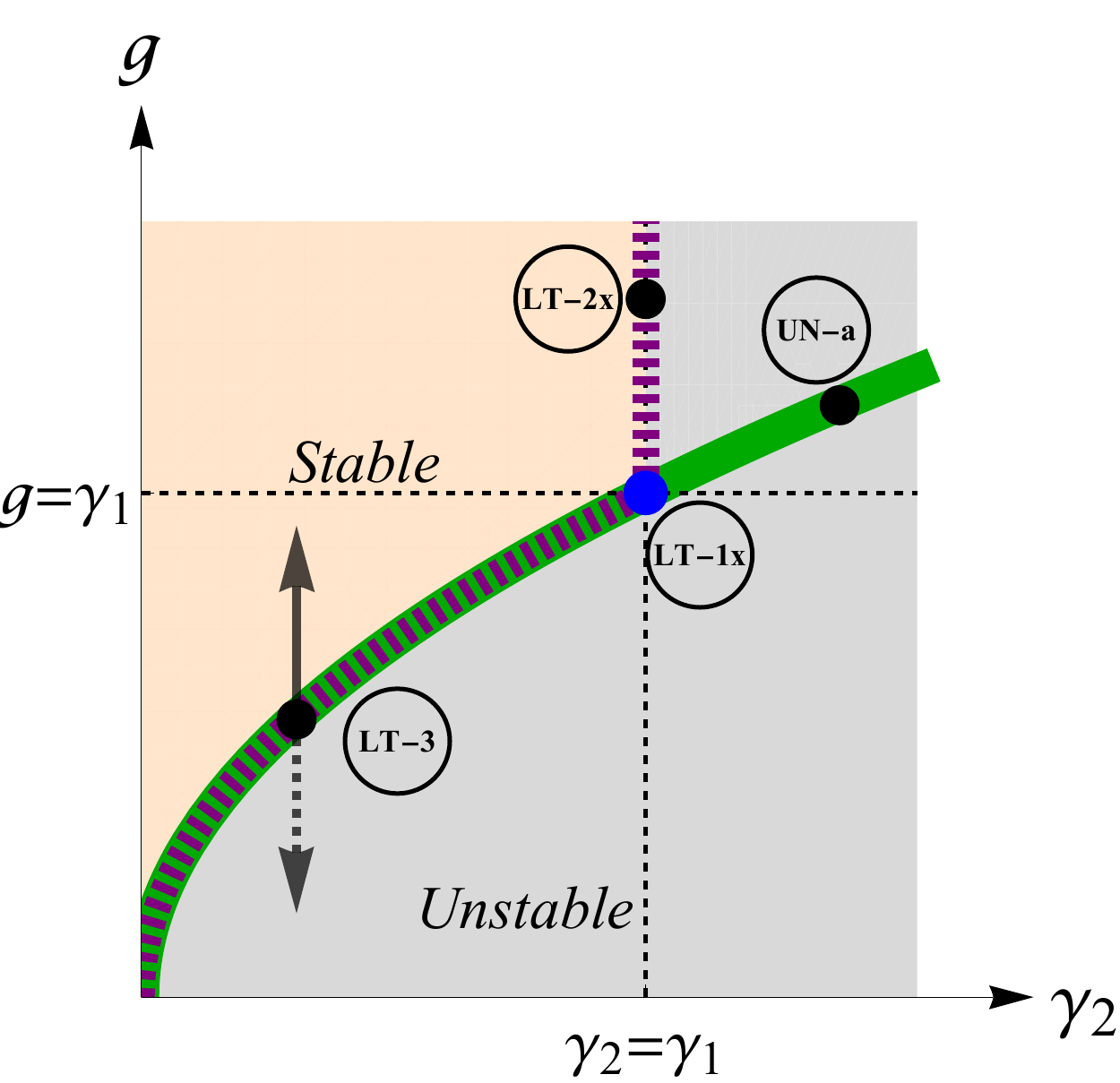}
	\caption{\textbf{Parameter space of coupling vs the driving rate} in the second cavity, $\{g,\gamma_2\}$, given some fixed loss rate $\gamma_1$ in the first cavity. From being stable the system becomes unstable (exhibits lasing) upon increasing sufficiently $\gamma_2$, as marked by \emph{light orange} and \emph{light grey} regions, respectively. The \emph{blue dot} corresponds to the only EP exhibited in this parameter space. The points labelled by LT-k constitute examples of operating at the lasing threshold (\emph{purple dashed line}), and pertain (not necessarily to scale) to the parameter values listed in \tabref{tab:EP_Sing_Bal}. The only new configuration LT-3 corresponds to the case UN-a in \tabref{tab:EP_Sing_Bal}, but with values of parameters $\gamma_{1}$ and $\gamma_{2}$ interchanged---we use the configuration LT-3 below to show that both perturbations into the stable and unstable phases (\emph{solid and dashed arrows}, respectively) away from the LT may lead to unbounded sensitivity.}
	\label{fig:lasing}
\end{figure}

\subsubsection{Steady-state solutions for the mean photon numbers in the cavities}
In order to have a better understanding of the interplay between the stability of the sensor dynamics and its sensitivity, we look at the steady-state solution of the average occupation number of the cavity modes $\hat{a}_1$ and $\hat{a}_2$. We adopt a simplistic approach based on the master equation of the reduced state $\rho$ of the two cavity modes. This approach can be shown to corroborate with more rigorous treatments accounting for the correlations being accumulated over time between cavity modes and the input fields, with the latter  possibly even in non-thermal states~\cite{NewProject}. 

In particular, we assume here that the reduced dynamics of the cavities can be modelled by the following master equation
\begin{align}
	\frac{d}{dt} \rho &= -\i [\hat{H}_S, \rho] + 2 \gamma_1 \left[ \hat{a}_1\rho \hat{a}_1^\dagger - \frac{1}{2}\left(\hat{a}_1^\dagger \hat{a}_1 \rho  + \rho \hat{a}_1^\dagger \hat{a}_1\right)\right] \nonumber 
	\\&+ 2 \gamma_2 \left[ \hat{a}_2^\dagger \rho \hat{a}_2 - \frac{1}{2}\left(\hat{a}_2  \hat{a}_2^\dagger \rho+ \rho \hat{a}_2 \hat{a}_2^\dagger \right)\right],
\end{align}
where the system Hamiltonian describing the two two cavities, $\hat{H}_S =\omega_1 \hat{a}_1^\dagger \hat{a}_1 + \omega_2 \hat{a}_2^\dagger \hat{a}_2 +  g (\hat{a}_1^\dagger \hat{a}_2 + \hat{a}_2^\dagger \hat{a}_1)$, is given as originally in \eqnref{eq:cavityHam}. 

A straightforward calculation then leads to
\begin{align}
	\partial_t \langle \hat{n}_{1} \rangle &= -2\gamma_{1} \langle\hat{n}_{1} \rangle +  2 g \langle \hat{J}_y \rangle  \label{eq:partialtn1} \\
	\partial_t \langle\hat{n}_{2} \rangle  &= 2\gamma_{2} \langle \hat{n}_{2} \rangle  -  2 g \langle\hat{J}_y \rangle  \label{eq:partialtn2} \\
	\partial_t  \langle\hat{J}_y \rangle  &= g (\langle\hat{n}_2\rangle - \langle\hat{n}_1\rangle)  + (\gamma_2 - \gamma_1) \langle\hat{J}_y\rangle \label{eq:partialtJy}\\
	&+ (\omega_1 - \omega_2) \langle \hat{J}_x \rangle \nonumber\\
	\partial_t  \langle\hat{J}_x \rangle  &= (\omega_2 - \omega_1) \langle \hat{J}_y \rangle   + (\gamma_2 - \gamma_1) \langle\hat{J}_x\rangle, \label{eq:partialtJx}
\end{align}
where $\hat{n}_{\ell} = \hat{a}_\ell^\dagger \hat{a}_\ell$ are the cavity occupation operators and $\hat{J}_x =  \frac{1}{2}(\hat{a}_1 \hat{a}_2^\dagger  + \hat{a}_2 \hat{a}_1^\dagger) $,  $\hat{J}_y = -\frac{\i}{2}(\hat{a}_1 \hat{a}_2^\dagger - \hat{a}_2 \hat{a}_1^\dagger)$ correspond to the relevant SU(2) generators in the Jordan-Schwinger representation. Moreover, since throughout this work we consider internal frequencies of the cavities to be equal, $\omega_1 = \omega_2 \eqqcolon \omega_0$, the last equation \eref{eq:partialtJx} in the above becomes redundant. Hence, by equating the r.h.s of Eqns.~(\ref{eq:partialtn1}-\ref{eq:partialtJy}) to zero and eliminating $\langle\hat{J}_y\rangle$, we obtain the steady-state solutions for the average photon number in each cavity as:
\begin{align}
	\left< \hat{n}_{1} \right>_\textrm{s.s.} &= \frac{g^2 \gamma_2}{(\gamma_1 - \gamma_2)(g^2 -\gamma_1 \gamma_2)}, \label{eq:n1avg}\\
	\left<\hat{n}_{2} \right>_\textrm{s.s.}  &= \frac{(g^2 -\gamma_1 \gamma_2) + g^2 \gamma_2}{(\gamma_1 - \gamma_2)(g^2 -\gamma_1 \gamma_2)}.\label{eq:n2avg}
\end{align}

Now, for the steady-state solution to exist, both $\langle \hat{n}_1 \rangle_\textrm{s.s.}$ and $\langle \hat{n}_2 \rangle_\textrm{s.s.}$ above must be non-negative and finite. From the perspective of increasing the driving rate $\gamma_2$ from zero, it becomes clear that the transition occurs at 
\begin{equation}
	\gamma_2^\textrm{LT}=\min\left\{\gamma_1,\frac{g^2}{\gamma_2}\right\},
	\label{eq:gamma2_LT}
\end{equation} 
at which both \eqnsref{eq:n1avg}{eq:n2avg} diverge. In particular, the above value defines the \emph{lasing threshold} (LT), below which, i.e.~for $\gamma_2<\gamma_2^\textrm{LT}$, the system is stable. Otherwise, i.e.~for $\gamma_2\ge\gamma_2^\textrm{LT}$, the system becomes unstable (UN) exhibiting the lasing phase. Let us note that the above conclusions are consistent with the argument presented in the main text based on the change of sign in the imaginary parts of the eigenvalues of the dynamical generator $\H$ (and similarly $\smat{M}$ in \eqnref{eq:Mgenerator}).

In \figref{fig:lasing}, depicting the parameter space of the coupling constant $g$ vs the driving rate $\gamma_2$ (for a given fixed value of the loss $\gamma_1$), the LT determined by \eqnref{eq:gamma2_LT} is marked with a dashed purple line (same as~\figref{fig:par_space}) that separates \emph{stable} and \emph{unstable} phases filled with light orange and light grey colours in \figref{fig:lasing}, respectively. 

\subsubsection{Singularity-induced sensing at the lasing threshold}
\label{app:sense_at_LT}
It is clear from \figref{fig:lasing} that when the intracavity coupling is relatively weak $g\le \gamma_1$ (as compared to the loss rate in the first cavity), the LT coincides with the singularity condition $\det\H=0$ (marked by a green solid curve in \figref{fig:lasing}, in a similar fashion to \figref{fig:par_space}), so that sensitivity to $\theta_0$-perturbations is then indeed unbounded when probing the system on resonance with $\omega=\omega_0$. Consistently, the configuration LT-1a in \figref{fig:Det0} (top-left subplot), for which $g\le \gamma_1$ in \tabref{tab:EP_Sing_Bal}, leads to a vanishing error $\qerrorS\theta_0\to0$ as $\theta_0\to0$. 

On the contrary, in the strong-coupling regime $g>\gamma_1$, the singularity is satisfied away from the LT in the unstable (UN) lasing phase (solid green and dashed purple lines in \figref{fig:lasing} then do not coincide). Hence, it is rather the configuration UN-a instead of LT-2b in \figref{fig:Det0} (with parameters specified in \tabref{tab:EP_Sing_Bal}) that leads to unbounded sensitivity for $\omega=\omega_0$. However, as we now show, for any point at the LT in the strong-coupling regime (i.e.~all the points along the vertical segment of the dashed purple line in \figref{fig:lasing}, including LT-2b) one can make the effective dynamical generator, $\smat{M}=(\omega-\omega_0)\Id-\smat{H}$ in \eqnref{eq:Mgenerator}, singular by tuning the probing frequency $\omega$.

\begin{figure}[t!]
	\centering
	\includegraphics[width=\columnwidth]{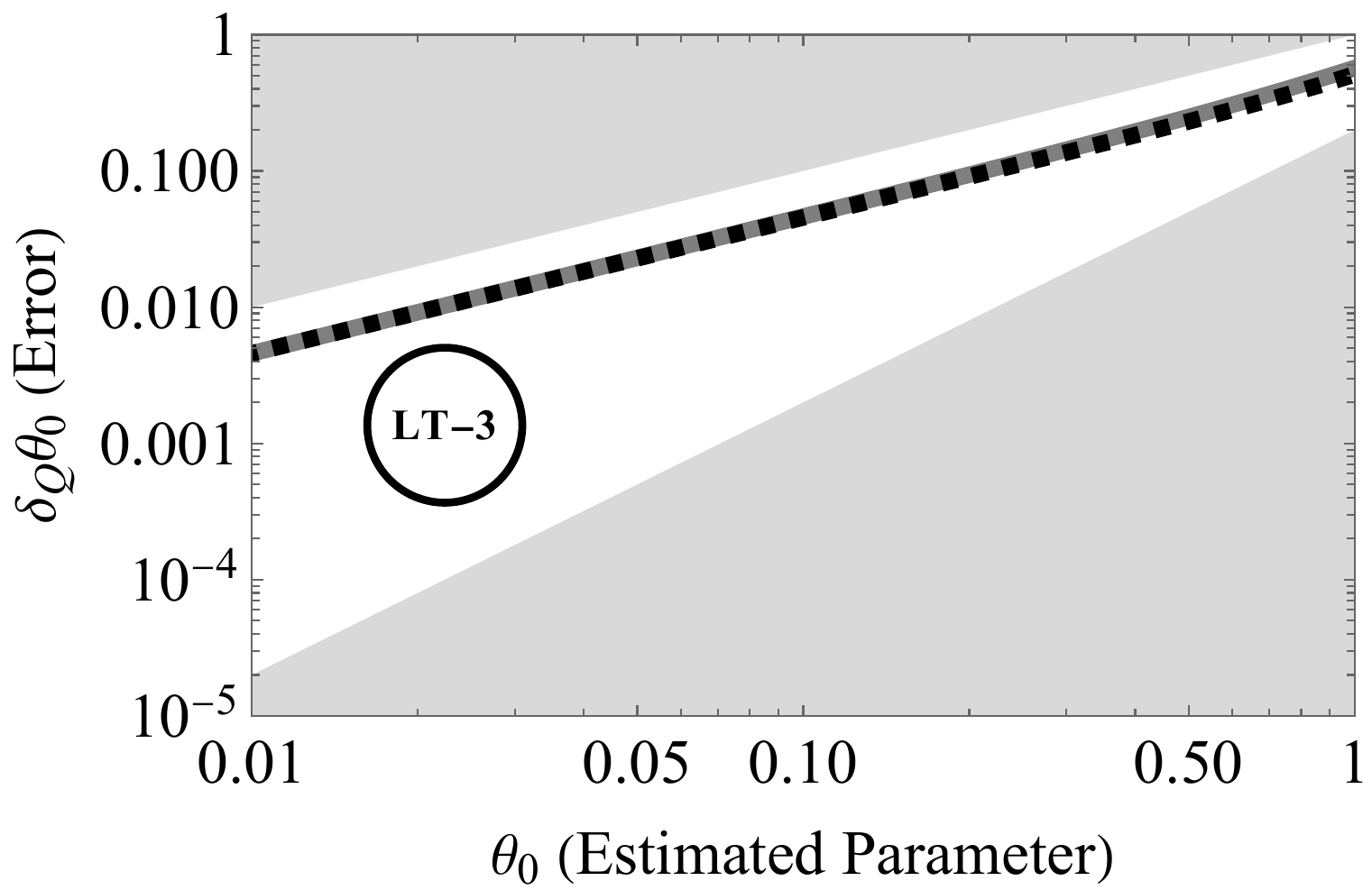}
	\caption{
		\textbf{Singularity-enhanced sensing at the lasing threshold} with the perturbation taking the system into either the stable or unstable phase. We consider the configuration marked as LT-3 in \figref{fig:lasing} and vary the coupling strength $g$, i.e.~$\mat{n}_0=\sigma_x$ in \eqnref{eq:response_function}. The error vanishes linearly, $\qerrorS\theta_0\propto\theta_0^{-1}$, matter whether $g\to g + \theta_0$ (solid grey) or $g\to g -\theta_{0}$ (dashed black) as $\theta_0\to0_+$. In the above, the system parameters are set to:~$\omega = \omega_0$, $g=2\gamma$, $\gamma_1 = 4 \gamma$, $\gamma_2 = \gamma$.
	}
	\label{fig:couplingPerturb}
\end{figure}

Defining the detuning of the probing frequency from the one of cavities as $\lambda\coloneqq\omega-\omega_0$, the condition for the dynamical generator \eqref{eq:Mgenerator} to be singular reads
\begin{align}
	\det\smat{M} = \det\!\left(\stransf{\mat{M}}\right) 
	& =  \left|\det \mat{M}\right|^2 \label{eq:detMequivalence}\\
	& =  \left|\det \!\left(\lambda\id-\H\right)\right|^2=0,	
	\label{eq:detM=0}
\end{align}
where $\mat{M}=\lambda\id-\H$ is the representation of the dynamical generator \eref{eq:Mgenerator} in the vector space of annihilation/creation operators rather than the phase space, see \eqnref{eq:phase_space_rep}, with corresponding determinants related as in \eqnref{eq:detMequivalence} due to: 
\begin{lemma}
	For any two square and real matrices of dimension $d$, $X,Y\in\mathcal{M}_{d,d}(\mathbb{R})$, the following equality holds
	\begin{equation}	
		\det
		\begin{pmatrix}
			X	& 	-Y \\
			Y	& 	X
		\end{pmatrix} 
		= \left|\det Z\right|^2,
		\label{eq:block_det}
	\end{equation}
	where $Z\coloneqq X+\i Y\in\mathcal{M}_{d,d}(\mathbb{C})$.
\end{lemma}
\begin{proof}
	Defining $I$ as the $d$-dimensional identity matrix, let us observe that
	\begin{equation}	
		\begin{pmatrix}
			X	& 	-Y \\
			Y	& 	X
		\end{pmatrix} 
		=\frac{1}{\sqrt{2}}\!
		\begin{pmatrix}
			I	& 	I \\
			\i I	& 	- \i I
		\end{pmatrix} 
		\,
		\begin{pmatrix}
			Z^*	& 	0 \\
			0	& 	Z
		\end{pmatrix} 
		\,
		\frac{1}{\sqrt{2}}\!
		\begin{pmatrix}
			I	& 	-\i I \\
			I	& 	\i I
		\end{pmatrix}, 
	\end{equation}
	so from multiplicativity of the determinant and its standard form for block matrices with commuting blocks, we obtain
	\begin{equation}	
		\det
		\begin{pmatrix}
			X	& 	-Y \\
			Y	& 	X
		\end{pmatrix} 
		= (-\i)^d \,\det(Z^*Z)\,(\i)^d=\det Z^* \det Z
	\end{equation}
	and \eqnref{eq:block_det} follows.
\end{proof}
Now, as \eqnref{eq:detM=0} is equivalent to just the characteristic equation of the ``bare'' dynamical generator  $\H$ (introduced in ~\eqnref{eq:sys_H}), $\det \smat{M}=0$ can be satisfied by choosing some detuning $\lambda\in\mathbb{R}$ \emph{if and only if} $\H$ possesses a real eigenvalue. The eigenvalues of $\H$ read $\lambda_\pm=\i\gamma_-\pm\sqrt{g^2-\gamma_+^2}$ (see ~\eqnref{eq:eval_evecs}), so it becomes clear that for $g^2\le\gamma_+^2$ this may only happen when one of them vanishes (i.e.~$\gamma_-^2=\gamma_+^2-g^2\Leftrightarrow g^2=\gamma_1\gamma_2$), so that $\det \H=0$ and we consistently recover the $\H$-singular condition when sensing on resonance. In contrast, if $g^2>\gamma_+^2$ then $\Im[\lambda_\pm]=\i\gamma_-$, so we must set $\gamma_-=0\;\Leftrightarrow\;\gamma_1=\gamma_2\eqqcolon\gamma$ (balanced case) to obtain real $\lambda_\pm=\pm\sqrt{g^2-\gamma^2}$ with $g>\gamma$. In particular, this way we recover exactly the aforementioned case of sensing at the LT in the strong-coupling regime ($\gamma_1=\gamma_2$ for $g>\gamma_1$), for which we must hence detune the probe by $\pm\sqrt{g^2-\gamma^2}$ to operate at the singularity. For instance, this is exactly the difference between the configurations LT-2b and LT-2a in \figref{fig:Det0} (middle subplots), where in the latter case we set the detuning to $\lambda=-\gamma$ given $g=\sqrt{2}\gamma$, see the corresponding parameters in \tabref{tab:EP_Sing_Bal}, in order to make the sensitivity unbounded. 

Finally, we demonstrate that the unbounded sensitivity may be exhibited when perturbing the system away from the singularity occurring at the LT, no matter whether the system is perturbed into the stable or the unstable phase. We consider an exemplary configuration in the weak-coupling regime denoted as LT-3 in \figref{fig:lasing} and assume $\theta_0$ to be a variation of the coupling strength (i.e.~$\mat{n}_0=\sigma_x$ in \eqnref{eq:response_function} and \tabref{tab:scaling}). As a result, when $\theta_0$ is positive(negative) the system becomes stable(unstable), as marked by the solid(dashed) arrow in \figref{fig:lasing}. In \figref{fig:couplingPerturb}, we explicitly show numerically that the error vanishes linearly, i.e.~$\qerrorS\theta_0\propto\theta_0^{-1}$ as $\theta_0\to0$, in either of the cases (the system parameters are fixed to:~$\omega = \omega_0$, $g=2\gamma$, $\gamma_1 = 4 \gamma$, $\gamma_2 = \gamma$; as in the case of UN-a in \tabref{tab:EP_Sing_Bal} but with the loss and driving rates interchanged).

\section{Multi-parameter perturbation sensing}
\label{app:multi_par}
%
For any dynamical transformation under which an input Gaussian state $\rho (\smat{S}_\mrm{in}, \smat{V}_\mrm{in})$ is mapped onto an output Gaussian state $\rho (\svec{S}(\parVec), \smat{V}(\parVec))$ with $\svec{S}(\parVec)$ and $\smat{V}(\parVec)$ respectively given by Eqs.~\eqref{eq:SparVec} and \eqref{eq:VparVec}, and the corresponding QFIM in \eqnref{eq:FQjk}.

\subsubsection{Sensing $\parVec$ at a singular point}
Now, if the response function admits a Laurent expansion $\smat{G}_{\theta_j}  = \sform \theta_j^{-s} \sum_{\ell=0}^r \theta_j^\ell\,\X_\ell$ with $s\in\mathbb{N}_+$, in particular, the SM expansion \eref{eq:SME}, then \eqnref{eq:FQjk} can be generally written as
\begin{align}\label{eq:FQ_ijgen}
	\FQ_{jk}	\approx  \theta_{j}^{-2s} \left( \sum_{p=0}^{2r}   \theta_{j}^p   \Tr\!\left[  \T^p_{jk}  \right] + \sum_{p=0}^{2r} \theta_{j}^p t^p_{jk} \right),
\end{align} 
with 
\begin{align}
	\T^p_{jk} &: = \sum_{\alpha = 0}^{p}  \smat{n}_{j} \smat{X}_\alpha  \smat{n}_{k} \smat{X}_{p-\alpha}+ \smat{V}_\mrm{in}^{-1} \smat{n}_{j} \smat{X}_\alpha  \smat{V}_\mrm{in} \smat{X}_{p-\alpha}^\T \smat{n}_{k}^\T,\\
	t^p_{jk} &:= \kappa^2 \sum_{\alpha = 0}^p  \svec{S}_\mrm{in}^\T   \smat{X}_\alpha^\T \smat{n}_{j}^\T \smat{V}_\mrm{in}^{-1} \smat{n}_{k} \smat{X}_{p-\alpha} \svec{S}_\mrm{in},
\end{align}
where $\smat{X}_\ell=0$ for any $\ell>r$ in the above.

\subsection{Impact of the nuisance parameter}
\label{app:nuisance_par}

In this work, we focus on the two-parameter $\parVec=\{\theta_0,\theta_1\}$ sensing scenario, in which $\theta_0$ is the primary parameter inducing the perturbation to be sensed, while $\theta_1$ represents a secondary \emph{nuisance} parameter, whose value is of no interest, but cannot be assumed to be apriori known, in contrast to all other constants ($g$, $\gamma_1$, $\gamma_2$, $\omega=\omega_0$) specifying the dynamics.

\subsubsection{Singularity-braking nuisance parameter}
\label{app:nonsing_pert}
Let us study first how the error in estimating the primary parameter $\theta_{0}$ is generally affected by the presence of a nuisance parameter $\theta_1$ that takes an otherwise singular generator $\smat{H}$ away from the singularity. This can be achieved in different ways via any $\smat{H}^\mrm{NS}_{\theta_1} = \smat{H} - \theta_1 \smat{n}_1$ such that $\det\smat{H}^\mrm{NS}_{\theta_1}\ne0$ for $\theta_1\neq0$. In general, the corresponding two-parameter response function  in~\eqnref{eq:Gomega} then reads
\begin{equation}
	\smat{G}_{\theta_0, \theta_1} =  \smat{J} (\theta_{0} \smat{n}_{0} - \smat{H}^\mrm{NS}_{\theta_1})^{-1} = \smat{J} (\theta_{0} \smat{n}_{0} + \theta_{1} \smat{n}_{1} - \smat{H})^{-1},
\end{equation}
but since $\smat{H}^\mrm{NS}_{\theta_{1}}$ is invertible for $\theta_{1} \ne 0$, one can use the Neumann series expansion \eref{eq:expnG} with $\svec{H}$ replaced by $\smat{H}^\mrm{NS}_{\theta_{1}}$ and the coefficients $\smat{g}_k$ being now a function of the nuisance parameter $\theta_1$.  Similarly, the covariance matrix allows for an expansion as in \eqnref{eq:expnV} with $\smat{v}_k$ depending now on $\theta_1$ via $\smat{g}_k$. Consequently, the entries of the noisy QFIM \eref{eq:FmatrixNoisyparVecApp} read:
\begin{align}
	\FQ_{00} 	&\approx  \frac{1}{2} \Tr\!\left[(\smat{V}^A_\mrm{in}   -  \smat{v}_0)^{-1} (-\smat{v}_1) (\smat{V}^A_\mrm{in}   -  \smat{v}_0)^{-1} (-\smat{v}_1) \right] \nonumber \\&+2 \kappa^2 \svec{S}_\mrm{in}^\T \smat{g}_1^\T (\smat{V}^A_\mrm{in}   -  \smat{v}_0)^{-1}  \smat{g}_1 \svec{S}_\mrm{in} +\mathcal{O}(\theta_0), \label{eq:F00NS}\\
	\FQ_{01} &\approx  \frac{1}{2} \Tr\!\left[(\smat{V}^A_\mrm{in}   -  \smat{v}_0)^{-1} (-\smat{v}_1) (\smat{V}^A_\mrm{in}   -  \smat{v}_0)^{-1} \left(- \partial_{\theta_1} \smat{v}_0 \right) \right] \nonumber \\&+2 \kappa^2 \svec{S}_\mrm{in}^\T \smat{g}_1^\T (\smat{V}^A_\mrm{in}   -  \smat{v}_0)^{-1}  (\partial_{\theta_1}\smat{g}_0) \svec{S}_\mrm{in} +\mathcal{O}(\theta_0).\label{eq:F01NS}
\end{align}
where the first expression is consistently equivalent to the constant $C$ in \eqnref{eq:C_const}, while the second one is now a consequence of the multi-parameter scenario. 

In contrast to the single-parameter scenario, the error in estimating the parameter $\theta_0$, while treating $\theta_1$ as a nuisance one, is generally lower-bounded by $\qerror\theta_0=\sqrt{[\FQ^{-1}]_{00}}$, where
\begin{align}\label{eq:error_multi}
	\left[\FQ^{-1} \right]_{00}  
	&=  \frac{\FQ_{11} }{ \FQ_{00}  \FQ_{11} -  \FQ_{01}  \FQ_{10}} =  \frac{1}{ \FQ_{00} - \frac{ \FQ_{01}  \FQ_{10}}{ \FQ_{11}}}.
\end{align}
Hence, in order for $\left[\FQ^{-1} \right]_{00} $ to vanish, $\FQ_{00} - \frac{ \FQ_{01}  \FQ_{10}}{ \FQ_{11}}$ must blow up. This is possibility is excluded by \eqnref{eq:F00NS} and \eqnref{eq:F01NS}, given that $\smat{n}_0$, $\smat{n}_1$, $\smat{H}_{\theta_1}^{-1}$ and $\smat{V}_\mrm{in}^{-1}$ all have a finite spectral radius. 

As a concrete example, we focus on $\smat{H}^\mrm{NS}_{\theta_1}$ stated in~\eqnref{eq:HS_HNS} with $\theta_1$ corresponding to a perturbation of the inter-cavity coupling strength $g$ in our system, depicted by the black arrow in~\figref{fig:par_space}. Taking the perturbation of the primary parameter to be generated by $\smat{n_0}=\Id$, we choose the input modes of both the probing and scattering channels to be in thermal states with zero average amplitude, i.e.~$\svec{S}_\mrm{in}=\svec{0}$, while holding an average number of photons $n_A$ and $n_B$, respectively. In such a case, we find the off-diagonal element \eref{eq:F01NS} vanishes, so that $\left[\FQ^{-1} \right]_{00}  = 1/\FQ_{00}$ with 
\begin{equation}
	\FQ_{00} \approx  \frac{2(2\bar{n}_A + \eta_2 \bar{n}_B) - \bar{n}_A(7\bar{n}_A + 3\eta_2 \bar{n}_B)}{c+2 c_{+} \theta_1 + c_{-} \theta_1^2 - 4 \bar{n}_{A}^2 \theta_1^3 + \bar{n}_A^2 \theta_1^4},
\end{equation}
where $\bar{n}_{A(B)} =  (1 + 2 n_{A(B)} )$, $c = \eta_1 \bar{n}_{B}(\bar{n}_{A} + \eta_2 \bar{n}_{B}) - \bar{n}_{A}(7\bar{n}_{A} + 3 \bar{n}_{B} \eta_2)$, $c_{\pm} =  \bar{n}_{A}(2\bar{n}_{A} \pm (\eta_1 + \eta_2) \bar{n}_{B})$, and we have set $\kappa =1$ for convenience. Thus, $\qerror\theta_0=\sqrt{\left[\FQ^{-1} \right]_{00}}$ is non-vanishing for any small $\theta_1\neq0$.

Now, let us also consider the very special case of $\theta_{1} = 0$, at which $\smat{H}^\mrm{NS}_{\theta_1=0} = \smat{H}$, so that it becomes singular. Hence, we must then resort to the SM expansion of the response function appearing in \eqnref{eq:FQjk}, i.e.~determine $\smat{G}_{\theta_0, \theta_1 =0} = \sform \theta^{-2} \sum_{k=0}^1 \theta_{0}^k \X_k$ that turns out to be identical to the single-parameter case discussed in \appref{app:SM}:~the pole is of order two, coefficient $\X_0$ is given by \eqnref{eq:X0_2mode}, and $\X_1 = \Id$. Assuming the inputs to be in zero-amplitude thermal states as above, \eqnref{eq:FQjk} then yields
\begin{align}
	\FQ_{00} &\approx \alpha_{00} \theta_{0}^{-4} + \beta_{00} \theta_{0}^{-2}, \label{eq:F00}\\
	\FQ_{01} &\approx \alpha_{01}  \theta_{0}^{-3},\\
	\FQ_{11} &\approx \alpha_{11}  \theta_{0}^{-4}  + \beta_{11} \theta_{0}^{-2}.
\end{align}
where
\begin{align}
	\alpha_{00} &= \frac{2(2 + 4 n_A + \pi_{1B}  + \pi_{2B} )^2}{(1 + 2 n_A + \pi_{1B} ) (1 + 2 n_A + \pi_{2B} )}, \quad \beta_{00} =4,\nonumber \\
	\alpha_{01} &= \alpha_{11} = 2 \left(   6 + \frac{1 + 2 n_A + \pi_{1B} }{1 + 2 n_A + \pi_{2B}} + \frac{1 + 2 n_A + \pi_{2B} }{1 + 2 n_A + \pi_{1B}} \right) \nonumber \\&= 12 + \beta_{11}.
\end{align}
Using the above expressions, we have that
\begin{equation}
	\frac{\FQ_{01}\FQ_{10}}{\FQ_{11}}=\frac{\FQ_{01}^2}{\FQ_{11}} 
	\approx \frac{\alpha_{01} \theta_{0}^{-6}}{\theta_{0}^{-4} + \frac{\beta_{11}}{\alpha_{01}} \theta_{0}^{-2}},
\end{equation}
which when plugged into \eqnref{eq:error_multi}, yields the quantum bound on the (multiparameter) estimation error: 
\begin{align}
	\qerror \theta_{0} 
	&= \sqrt{[\FQ^{-1}]_{00}} \nonumber \\
	&\approx \left( \alpha_{00} \theta_{0}^{-4}  + \beta_{00}\theta_{0}^{-2} - \frac{\alpha_{01} \theta_{0}^{-6}}{\theta_{0}^{-4} + \frac{\beta_{11}}{\alpha_{01}} \theta_{0}^{-2}}  \right)^{-1/2} \nonumber \\
	&= \frac{1}{\alpha_{00}^{1/2}} \theta_{0}^2 -\frac{(\beta_{00} - \alpha_{01})}{2\alpha_{00}^{3/2}}  \theta_{0}^{4}+O(\theta_0^6).
	\label{eq:qerrorNS}
\end{align}

On the other hand, in the single-parameter scenario, in which $\theta_1=0$ is assumed to be perfectly known and not fluctuating, the quantum bound on the estimation error follows from \eqnref{eq:F00} and reads
\begin{align}
	\qerrorS \theta_0 
	&= 1/\sqrt{\FQ_{00}} \nonumber\\
	&\approx  (\alpha_{00} \theta_{0}^{-4} + \beta_{00} \theta_{0}^{-2})^{-1/2} \nonumber \\
	&= \frac{1}{\alpha_{00}^{1/2}} \theta_{0}^{2} - \frac{\beta_{00}}{2 \alpha_{00}^{3/2}} \theta_{0}^4 + O(\theta_0^6).
	\label{eq:qerrorSNS}
\end{align}
Hence, as $\qerror \theta_0 = \qerrorS \theta_{0} + O(\theta_0^4)$, it follows that despite the singular behaviour with error vanishing quadratically as $\theta_0\to0$, the impact of the nuisance parameter becomes negligible for small enough perturbations.

\subsubsection{Singularity-preserving nuisance parameter}
\label{app:sing_pert}
Considering now the perturbation $\theta_{1}$ in~\eqnref{eq:HS_HNS} such that $\HS$ is singular for all values of $\theta_{1}$ (including $\theta_{1} =0$), we must compute the SM expansion of the response function that is valid for any $\theta_{1}$. Following the steps described in \appref{app:SM}, one finds that $\smat{G}_{\theta_{0}, \theta_{1}} = \sform \theta_{0}^{-2} \sum_{k=0}^1 \theta_{0}^k \X_k$, where
\begin{equation}\label{eq:X0}
	\X_0 =   \begin{pmatrix}
		0       & 1-\theta_{1} & 1-\theta_{1}  & 0      \\
		1-\theta_{1}  &  0     &   0     & -1-\theta_{1} \\
		-1-\theta_{1} &  0     &   0     & 1-\theta_{1}   \\
		0       & 1-\theta_{1} &  1-\theta_{1} & 0
	\end{pmatrix}
\end{equation}
and $\X_1 = \Id$.
Hence, substituting the above SM expansion into \eqnref{eq:FQjk}, we obtain the corresponding elements of the QFIM:
\begin{align}
	\FQ_{00} &\approx \alpha \theta_0^{-4}    +  2\beta \theta_0^{-3}  + \gamma \theta_0^{-2}, \label{eq:F00S}\\
	\FQ_{01} &\approx \alpha \theta_0^{-3}+\beta \theta_0^{-2},\\
	\FQ_{11} &\approx  \alpha \theta_0^{-2},
\end{align}
where  $\alpha:=\Tr\!\left[    \smat{V}_\mrm{in}^{-1} \X_0 \smat{V}_\mrm{in} \X_0^\T \right]  +  \langle \X_{0}^\T  \smat{V}_\mrm{in}^{-1}  \X_0  + \X_{0}  \smat{V}_\mrm{in}^{-1}  \X_0^\T \rangle$, $\beta :=  \langle \X_0^\T  \smat{V}_\mrm{in}^{-1} \X_1 + \X_1^\T  \smat{V}_\mrm{in}^{-1} \X_0 \rangle$, and $\gamma = 8  + 2 \langle \X_1^\T  \smat{V}_\mrm{in}^{-1} \X_1 \rangle $ with  $\langle \bullet  \rangle = \smat{S}_\mrm{in}^\T \bullet \smat{S}_\mrm{in}$. From positivity of covariance matrices (and their inverses) it follows that $\alpha,\gamma\ge0$ and $\alpha\gamma\ge\beta^2$.

The above elements of the QFIM allow us to determine the quantum bound on the $\theta_0$-estimation error with $\theta_1$ constituting a nuisance parameter, i.e.:
\begin{align}
	\qerror\theta_{0} 
	& =\sqrt{[\FQ^{-1}]_{00}}
	\nonumber \\
	&\approx  
	\left(\alpha \theta_0^{-4}    +  2\beta \theta_0^{-3}  + \gamma \theta_0^{-2} - \frac{ \left(\alpha \theta_0^{-3}+\beta \theta_0^{-2} \right)^2}{\alpha \theta_0^{-2}}\right)^{-\frac{1}{2}}
	\nonumber\\
	&=
	\left(\gamma-\frac{\beta^2}{\alpha}\right)^{-\frac{1}{2}}\,\theta_0
	\label{eq:qerrorS}
\end{align}
which, in contrast to \eqnref{eq:qerrorNS}, scales now linearly with $\theta_0$.

Crucially, the presence of the nuisance parameter precludes now the quadratic scaling of the error observed in the single-parameter setting. Recall that when $\theta_1$ is perfectly known, the quantum bound on the estimation error can be evaluated analogously to \eqnref{eq:qerrorSNS}, being dictated by \eqnref{eq:F00S}, i.e:
\begin{align}
	\qerrorS \theta_0 
	&= 1/\sqrt{\FQ_{00}} \nonumber\\
	&\approx  
	\left(\alpha \theta_0^{-4}    +  2\beta \theta_0^{-3}  + \gamma \theta_0^{-2}\right)^{-\frac{1}{2}} \nonumber\\
	&=
	\alpha^{-1/2}\,\theta_0^{2}-\alpha^{-3/2}\beta\,\theta_0^{3}+O(\theta_0^4),
\end{align}
so that it follows the $\theta_0^2$-scaling for $\theta_0\ll1$---corresponding to the already discussed scenario of sensing a two-mode symmetric perturbation of the cavity frequencies for a singular $\smat{H}$, which yields a pole of order two ($s=2$), as stated in the first row of \tabref{tab:scaling} in \appref{app:SM}.

Note that in \figref{fig:CRB}, when computing the exact estimation errors $\qerror \theta_0$/$\qerrorS \theta_0$ numerically, we further assume the input modes of both the probing and scattering channels to be in thermal states with zero average amplitude and an average number of photons, $n_A$ and $n_B$, respectively. As a result, the above coefficients become $\beta=0$, $\gamma=8$ and
\begin{equation}
	\alpha(\theta_1) = \frac{2(2 + 4n_A + \pi_{1B} + \pi_{2B})^2 (1- \theta_{1})^2}{(1 + 2 n_A + \pi_{1B}) (1 + 2 n_A + \pi_{2B})},
\end{equation}
where $\pi_{1B} = \eta_1 (1 + 2 n_B )$ and $\pi_{2B} = \eta_2 (1 + 2 n_B )$ and $\alpha(\theta_1)$ depends explicitly on the nuisance parameter $\theta_1$ (whose variations correspond then to red arrows in \figref{fig:par_space}a), we present for the two values $\theta_0=0$ and $\theta_0=0.25$, i.e.~both when the secondary perturbation is vanishing or non-vanishing---recall that $\det\HS=0$ here crucially for any $\theta_1$. Although we are primarily concerned here with perturbations such that $\theta_1\ll1$, the above is true for any $|\theta_1|<1$, while at the special value of $\theta_1=1$ one would have to recompute the SM expansion---$\X_0=\smat{0}$ at $\theta_1=1$ in \eqnref{eq:X0}.

\end{document}